\documentclass{article}

\usepackage{arxiv}

\usepackage[utf8]{inputenc}             
\usepackage[T1]{fontenc}                
\usepackage[colorlinks]{hyperref}       
\usepackage[dvipsnames, table]{xcolor}
\usepackage{url}                        
\usepackage{booktabs}                   
\usepackage{colortbl}                   
\usepackage{array}                      
\usepackage{makecell}                   
\usepackage{amsfonts}                   
\usepackage{microtype}                  
\usepackage{float}                      
\usepackage{amsmath}                    
\usepackage{amsthm}                     
\usepackage{graphicx}                   
\usepackage{nameref}                    
\usepackage{siunitx}                    

\usepackage[labelformat=simple, singlelinecheck=false, labelfont={bf,large}]{subcaption}

\newtheorem{proposition}{Proposition}

\newcommand{\comment}[1]{}

\usepackage[numbers]{natbib}
\newcolumntype{P}[1]{>{\centering\arraybackslash}p{#1}}
\newcolumntype{M}[1]{>{\centering\arraybackslash}m{#1}}

\hypersetup{
linkcolor=BrickRed,
citecolor=Green,
filecolor=Mulberry,
urlcolor=NavyBlue,
menucolor=BrickRed,
runcolor=Mulberry,
linkbordercolor=BrickRed,
citebordercolor=Green,
filebordercolor=Mulberry,
urlbordercolor=NavyBlue,
menubordercolor=BrickRed,
runbordercolor=Mulberry
}

\title{Boolean function metrics can assist modelers to check and choose logical rules}

\author{John Zobolas\textsuperscript{1,*}, Pedro T. Monteiro\textsuperscript{2,3}, Martin Kuiper\textsuperscript{1} and Åsmund Flobak\textsuperscript{4,5} \\
{\small \textsuperscript{1}Department of Biology, Norwegian University of Science and Technology (NTNU), Trondheim, Norway} \\
{\small \textsuperscript{2}Department of Computer Science and Engineering, Instituto Superior Técnico (IST) - Universidade de } \\ {\small Lisboa, Lisbon, Portugal} \\
{\small \textsuperscript{3}INESC-ID, Lisbon, Portugal} \\
{\small \textsuperscript{4}Department of Clinical and Molecular Medicine, Norwegian University of Science and Technology (NTNU),}\\ {\small Trondheim, Norway} \\
{\small \textsuperscript{5}The Cancer Clinic, St. Olav’s Hospital, Trondheim, Norway} \\
{\small \textsuperscript{*}{\small To whom correspondence should be addressed.}}
}

\begin{document}

\maketitle
\begin{abstract}
Computational models of biological processes provide one of the most powerful methods for a detailed analysis of the mechanisms that drive the behavior of complex systems.
Logic-based modeling has enhanced our understanding and interpretation of those systems.
Defining rules that determine how the output activity of biological entities is regulated by their respective inputs has proven to be challenging. 
Partly this is because of the inherent noise in data that allows multiple model parameterizations to fit the experimental observations, but some of it is also due to the fact that models become increasingly larger, making the use of automated tools to assemble the underlying rules indispensable.

We present several Boolean function metrics that provide modelers with the appropriate framework to analyze the impact of a particular model parameterization. 
We demonstrate the link between a semantic characterization of a Boolean function and its consistency with the model’s underlying regulatory structure. 
We further define the properties that outline such consistency and show that several of the Boolean functions under study violate them, questioning their biological plausibility and subsequent use.
We also illustrate that regulatory functions can have major differences with regard to their asymptotic output behavior, with some of them being biased towards specific Boolean outcomes when others are dependent on the ratio between activating and inhibitory regulators.

Application results show that in a specific signaling cancer network, the function bias can be used to guide the choice of logical operators for a model that matches data observations.
Moreover, graph analysis indicates that the standardized Boolean function bias becomes more prominent with increasing numbers of regulators, confirming the fact that rule specification can effectively determine regulatory outcome despite the complex dynamics of biological networks.
\end{abstract}

\keywords{Boolean regulatory networks \and Boolean functions \and Truth Density \and Bias \and Complexity}


\section{Introduction}

The understanding of biological processes has been greatly stimulated by systems biology approaches \cite{Kitano2002,Chuang2010,Apweiler2018}.
The integration of mathematical models with the underlying biological knowledge and empirical observations can help us observe emergent systems properties, test new hypotheses, enhance the interpretability of the studied systems and guide innovations in areas such as medicine and drug discovery \cite{Aldridge2006}.
While multiple mathematical modeling frameworks exist, the scarcity of experimental data and the challenges posed by the development of quantitative large-scale biological networks, has favoured the simplicity and intuitiveness of more qualitative approaches, such as logic-based modeling \cite{Morris2010}.

At the heart of the mathematical representation of molecular biological networks lies the concept of regulation.
Regulation of activity, typically by changing the modification state, location or concentration of a biological entity, is a process which can be expressed by a mathematical function that combines the various regulatory inputs that affect the target, with a logic that describes how these regulators are integrated.
In Boolean logic-based modeling, the regulatory inputs are entities which can be expressed in two states: active ($1$) or inactive ($0$).
These entities are combined with logical rules to derive the \textit{Boolean regulatory function} (BRF) of the target entity.
For every possible regulatory input (combination of $0$ and $1$’s) the BRF will produce the end regulatory product, which is the activity of the target ($0$ or $1$).

The construction of a Boolean computational model starts with the assembly of information from literature and experimental observations, in the form of a Prior Knowledge Network (PKN), i.e. a list of network entities and their causal interactions (positive or negative) \cite{Traynard2017, Toure2020}.
The use of a PKN for accurate representation of biological reality and subsequent analysis and simulation requires the definition of the model formalism.
This is one of the most important steps in dynamical modeling since it directly translates to the choice of BRFs, i.e. the logical rules that together with the regulators define the activity state of each network target \cite{Wang2012}.
There have been several approaches related to the choice of BRFs, from using a standardized format \cite{Mendoza2006}, to automatically generating all possible BRFs compatible with the PKN and calibrating the rules in order to fit perturbation data \cite{SaezRodriguez2009, Videla2016, Gjerga2020}.
State-of-the-art approaches involve the automated construction of large-scale logical networks by inferring the logical rules from the topology and semantics of molecular interaction maps \cite{Aghamiri2020}.

Regardless of how a logical model is constructed, it has been shown in practice that expert curation, i.e. the manual fine-tuning of the logical rules to fit experimental data, can result in highly predictive models \cite{Flobak2015, Niederdorfer2020}, yet this is not trivially obtained with automatically constructed networks \cite{Bekkar2018}.
Because of the large function space complemented with a sparsity of observations and inherent noise in existing data, there is a wide range of plausible BRFs.
Thus, it is crucial to properly define function characteristics that can guide the modeler to a more informed function choice.
Our work is focused on explicating some of these metrics and using them to show for example which BRFs can be discarded due to biological inconsistencies with the underlying regulatory topology and which are biased towards specific Boolean outcomes.

The paper is structured as follows: Section \ref{sec:background} provides a list of notations and definitions to be used later in the text.
In Section \ref{sec:DNF}, we discuss the benefits of using the equivalent disjunctive normal form of a Boolean function to delineate its biological interpretability.
In Section \ref{sec:consistency}, we provide a set of properties that characterize the Boolean functions that are consistent with a given regulatory topology and show that several functions under study violate them.
In Section \ref{sec:truth-density}, we present the truth density metric as a means to evaluate if a Boolean function is biased or balanced with increasing number of regulators.
We also discuss the asymptotic properties of different functions relating to the ratio between activators and inhibitors.
Lastly, in Section \ref{sec:network-analysis}, we present evidence that the standardized Boolean functions are indeed biased and show how modelers can exploit such information for their own benefit.
The results are demonstrated in Boolean models derived from a cancer signaling network as well as from scale-free topologies that are applicable to most biological networks.
We close the paper with some discussion points in Section \ref{sec:discuss} and directions for future research in Section \ref{sec:future}.

\section{Background \label{sec:background}}

\subsection{Boolean regulatory functions}

\textit{Boolean regulatory functions} (BRFs) are Boolean functions used in the context of biological networks and modeling.
A mathematical description of such a function associates the activity output of a target biological entity with the Boolean input values of $n$ variables (the \textit{regulators}), such that $f_{BRF}:\{0,1\}^n \rightarrow \{0,1\}$.
Thus, the target's output state is binary, i.e. either 0 ($False$, denoting an inactive or inhibited state) or $1$ ($True$, indicating an active state).

One intuitive representation of a Boolean function is its \textit{truth table}, which is a list of all possible Boolean input configurations of the $n$ regulators along with their associated function output.
Since every regulator can be assigned two possible values ($0$ and $1$), the total number of input configurations (i.e. rows) in a truth table is $2^n$.
For example, a Boolean function $f(x_1,x_2,x_3)$ with $3$ regulators has a total of $2^3=8$ rows in its corresponding truth table, starting from the input configuration $(0,0,0)$ and ending with $(1,1,1)$ (Table \ref{tab:table1}).

The total number of BRFs with $n$ regulators is $2^{2^n}$ since for each of the $2^n$ input configurations (i.e. rows of the truth table) there can be two possible function outcomes ($0$ or $1$).
For example, with $3$ regulators and a total of $8$ rows in the truth table, that would be a total of $2^8=256$ functions, three of which are shown in Table \ref{tab:table1}.

\subsection{Disjunctive normal form}

The most frequently used form of a Boolean function is its analytical expression, where variables are connected with logical operators such as AND ($\land$), OR ($\lor$), NOT ($\lnot$), XOR ($\oplus$), etc. and the output of the function is calculated using basic Boolean algebra.
In Table \ref{tab:table1} for example, we provide the analytical forms for the functions $f_1$ and $f_2$.
Note that there can be multiple analytical forms that essentially compute the same function, e.g. another form of the $f_1$ function is $f_1^{'}=(\lnot x_1 \land x_2 \land \lnot x_3) \lor (x_1 \land \lnot x_2 \land \lnot x_3) \lor (x_1 \land x_2 \land \lnot x_3)$.

\vspace{20pt}
\begin{table}[h!]
  \centering
  \begin{tabular}{cccccc}
    \toprule
    \multicolumn{3}{c}{Truth Table} & \multicolumn{3}{c}{Boolean functions} \\
    \cmidrule(r){1-3}
    \cmidrule(r){4-6}
    $x_1$ & $x_2$ & $x_3$ & $f_1=(x_1 \land \lnot x_3) \lor (x_2 \land \lnot x_3)$ & $f_2=x_1 \lor (\lnot x_2 \land \lnot x_3)$ & $f_3 = 1$ \\
    \midrule
    0&0&0&0&1&1 \\
    0&0&1&0&0&1 \\
    0&1&0&1&0&1 \\
    0&1&1&0&0&1 \\
    1&0&0&1&1&1 \\
    1&0&1&0&1&1 \\
    1&1&0&1&1&1 \\
    1&1&1&0&1&1 \\
    \bottomrule
  \end{tabular}
  \bigskip
  \caption{Truth table of three Boolean functions with three input variables $x_1, x_2$ and $x_3$. Functions $f_1$ and $f_2$ are expressed in disjunctive normal form (DNF) with the minimum possible number of terms. $f_3$ is a tautology.}
  \label{tab:table1}
\end{table}

This brings us to the notion of a general form which could be used to define useful metrics common to all Boolean functions (e.g. complexity), as well as the need to provide minimal forms based on specific criteria.
For example, a more compact function form enhances readability, which can be seen by comparing $f_1$ with $f_1^{'}$.

Every Boolean function can be represented in a \textit{disjunctive normal form} (DNF), requiring only AND ($\land$), OR ($\lor$) and NOT ($\lnot$) operators as building blocks.
In such a representation, \textit{literals}, which are variables (e.g. positive literal $x$) or their logical negations (e.g. negative literal $\text{NOT }x$), are connected by AND's, producing \textit{terms}, which are then in turn connected by OR's \cite{crama2011boolean}.
For example, every function in Table \ref{tab:table1} is expressed in DNF, while the Boolean expressions $\lnot (x_1 \lor x_2)$ and $\lnot (x_1\land x_2)\lor x_3$ are not.
Note that a Boolean function can have multiple DNF formulations.

\subsection{Link operator functions}

We consider the class of BRFs that partitions the input regulators to two sets: the set of positive regulators (\textit{activators}) and the set of negative regulators (\textit{inhibitors}).
Let $f$ be such a Boolean function $f_{BRF}(x,y):\{0,1\}^n \rightarrow \{0,1\}$, with $m \ge 1$ activators $x=\{x_i\}_{i=1}^{m}$ and $k \ge 1$ inhibitors $y=\{y_j\}_{j=1}^{k}$, that is a total of $n = m + k$ regulators.
The \textit{link operator} BRFs have an analytical formula which places the two distinct types of regulators in two separate expressions, while connecting them with a special logical operator that we call a \textit{link operator}.
An example of such a function that has been used extensively in the logical modeling literature is the standardized BRF formula with the “AND-NOT” link operator \cite{Mendoza2006}:

\begin{equation} \label{eq:and-not}
    f_{AND-NOT}(x,y) = \left(\bigvee_{i=1}^{m} x_i\right) \land \lnot \left(\bigvee_{j=1}^{k} y_j\right)
\end{equation}

A variation of the above function is the “OR-NOT” link operator function:

\begin{equation} \label{eq:or-not}
    f_{OR-NOT}(x,y) = \left(\bigvee_{i=1}^{m} x_i\right) \lor \lnot \left(\bigvee_{j=1}^{k} y_j\right)  
\end{equation}

Note that the presence of the link operator is what forces the condition $m,k \ge 1$ (at least one regulator in each category).
For the rest of this work, we will not consider BRFs with only one type of regulator, since these can be represented by simple logical functions without loss of biological consistency.
Following the notation introduced in Mendoza et al. \cite{Mendoza2006}, in the case of only positive regulators, the presence of at least one activator makes the target active, i.e. $f(x)=\bigvee_{i=1}^{m} x_i$.
In the case of only inhibitory regulators, the presence of at least one inhibitor is sufficient to make the target inactive, i.e. $f(y)=\lnot \bigvee_{j=1}^{k} y_j=\bigwedge_{j=1}^{k} \lnot y_j$.

Borrowing notation from circuit theory, we will also use other link operators like the “NAND”, “NOR”, “XNOR” gates, with or without the “NOT” symbol in front.
Note that the logical operator used to connect the same type of regulators (e.g. the activators) is usually OR, but other operators could be used as well.

Another link operator function that we will consider in this work is the “Pairs” function:

\begin{equation} \label{eq:pairs}
    f_{Pairs}(x,y) = \bigvee_{\forall (i,j)}^{m,k}(x_i\land \lnot y_j)
\end{equation}

The intuition behind the name is derived from the fact that the function will return $True$ if there is at least one pair of regulators consisting of a present activator and an absent inhibitor.
For a formulation of the “Pairs” function that is consistent with the link operator terminology as defined above, see (Eq. \ref{eq:pairs-cnf}).

\newpage
\subsection{Threshold functions}

\textit{Threshold functions} are a special type of Boolean functions, the output of which depends on the condition that the sum of (possibly weighted) activities of the input regulators surpasses a given \textit{threshold} value \cite{McCulloch1943, Hopfield1982}.

In this work we will consider two simple threshold functions, which both output $True$ when the number of present activators is larger than the number of present inhibitors.
As such, the activities of the positive and negative regulators are combined in an \textit{additive} manner, with their respective assigned weights set to $\pm1$ and the threshold parameter to $0$, formulating thus a \textit{majority rule} which defines the value of the function \cite{Bornholdt2008, Chaouiya2013}.
These functions differ with regards to their output when there is balance between the activities of the positive and negative regulators: the first outputs $1$ (the activators “win”) while the second outputs $0$ (the inhibitors “win”):

\begin{align} \label{eq:thres-act}
    f_{Act-win}(x,y) &= \begin{cases}
    1, & \sum_{i=1}^{m} x_i \ge \sum_{j=1}^{k} y_j \\
    0, & \text{otherwise} \end{cases}
    \intertext{}
    f_{Inh-win}(x,y) &= \begin{cases} \label{eq:thres-inh}
    1, & \sum_{i=1}^{m} x_i > \sum_{j=1}^{k} y_j \\
    0, & \text{otherwise} \end{cases}
\end{align}

\section{Disjunctive Normal Form unmasks biological interpretation} \label{sec:DNF}

\subsection{Interpretability issues in Boolean modeling}

Two main features make Boolean modeling attractive to users. 
First, transforming conditions for the activation or inhibition of a target biological entity to Boolean equations is a relatively easy task using a qualitative, logic-based modeling formalism. 
Second, the reverse is also true, i.e. Boolean equations can be more interpretable and closer to a simplified description of biological reality that “makes sense” than the use of other kinds of formalisms (e.g. kinetic modeling). 
For example, consider the simple case of a target entity, which is regulated by one positive regulator $x_1$ and one negative regulator $y_1$. 
The use of the “AND-NOT” link operator function in this case (Eq. \ref{eq:and-not}) is very easy to understand and interpret since the formula directly connects to the underlying biology.
Thus, the mathematical formulation is simply written as $f_{AND-NOT}=x_1 \text{ AND NOT } y_1$, while the modeler reads “the target becomes active when $x_1$ (the activator) is present and $y_1$ (the inhibitor) absent”.

Issues start arising when considering the \textit{interpretability} of such Boolean expressions in cases where a larger number of regulators act on a target, e.g. in a more complex scenario with three positive ($x_1,x_2,x_3$) and three negative ($y_1,y_2,y_3$) regulators, the mathematical formulation expressing the target's activity output can be easily written using the link operator function form, as $f_{AND-NOT}=(x_1 \text{ OR }x_2\text{ OR }x_3) \text{ AND NOT } (y_1\text{ OR }y_2\text{ OR }y_3)$.
A modeler could read this as “the target becomes active when at least one activator is present, and all of its inhibitory regulators are absent”, but a precise semantic description that explicates the conditions under which the target gets activated, can in general be difficult to assess.
A similar issue arises when reflecting on the use of a different link operator instead of the standard “AND-NOT” or even of an entirely different regulatory function, for which the biological interpretation might be difficult to derive from the expression itself.

\renewcommand{\arraystretch}{1.4} 
\begin{table}[H] 
  \scriptsize
  \resizebox{\textwidth}{!}{
  \begin{tabular}{|M{5cm}|M{4.7cm}|M{3.4cm}|M{1.4cm}|c|}
  \hline
     \thead{BRF (standard form)} & 
     \thead{BRF (CDNF)} & 
     \thead{Biological Interpretation} & 
     \thead{Consistent} & 
     \thead{Complexity} \\
     \hline
        $(x_1 \text{ OR }x_2)\textbf{ NOR }(y_1 \text{ OR }y_2)$ & 
        $\text{NOT }x_1\text{ AND NOT }x_2\text{ AND}$ \newline $\text{NOT }y_1\text{ AND NOT }y_2$ \hphantom{mmmm} & 
        Absence of all regulators & 
        \cellcolor[HTML]{FF4136} NO & 
        $1$ (always) \\
    \hline
        $(x_1 \text{ OR }x_2)\textbf{ NAND }(y_1 \text{ OR }y_2)$ & 
        $(\text{NOT }x_1\text{ AND NOT }x_2)\textbf{ OR}$ \newline $(\text{NOT }y_1\text{ AND NOT }y_2)$ \hphantom{mmm} &
        Absence of all activators \textbf{or} absence of all inhibitors &
        \cellcolor[HTML]{FF4136} NO & 
        $2$ (always) \\
    \hline
        $(x_1 \text{ OR }x_2)\textbf{ AND NOT }(y_1 \text{ OR }y_2)$ \newline \newline “AND-NOT” (Eq. \ref{eq:and-not}) & 
        $(x_1\text{ AND NOT }y_1\text{ AND NOT }y_2)\textbf{ OR}$ \newline $(x_2\text{ AND NOT }y_1\text{ AND NOT }y_2)$ \hphantom{mm} &
        Presence of at least one activator \textbf{and} absence of all inhibitors &
        \cellcolor[HTML]{2ECC40} YES & 
        $2$ ($m$) \\
    \hline
        $(x_1 \text{ OR }x_2)\textbf{ NOR NOT }(y_1 \text{ OR }y_2)$ & 
        $(y_1\text{ AND NOT }x_1\text{ AND NOT }x_2)\textbf{ OR}$ \newline $(y_2\text{ AND NOT }x_1\text{ AND NOT }x_2)$ \hphantom{mm} &
        Presence of at least one inhibitor \textbf{and}  absence of all activators &
        \cellcolor[HTML]{FF4136} NO & 
        $2$ ($k$) \\
    \hline
        $(x_1 \text{ OR }x_2)\textbf{ OR NOT }(y_1 \text{ OR }y_2)$ \newline \newline “OR-NOT” (Eq. \ref{eq:or-not}) &
        \hphantom{mm} $x_1\textbf{ OR }x_2\textbf{ OR}$ \newline $(\text{NOT }y_1\text{ AND}$ $\text{NOT }y_2)$ &
        Presence of any activator \textbf{or} absence of all inhibitors &
        \cellcolor[HTML]{2ECC40} YES &
        $3$ ($m+1$) \\
    \hline
        $(x_1 \text{ OR }x_2)\textbf{ NAND NOT }(y_1 \text{ OR }y_2)$ & 
        \hphantom{mm} $y_1\textbf{ OR }y_2\textbf{ OR}$ \newline $(\text{NOT }x_1\text{ AND NOT }x_2)$ &
        Presence of any inhibitor \textbf{or} absence of all activators &
        \cellcolor[HTML]{FF4136} NO & 
        $3$ ($k+1$) \\
    \hline
        $(x_1 \text{ OR }x_2)\textbf{ XOR }(y_1 \text{ OR }y_2)$ & 
        $(x_1\text{ AND NOT }y_1\text{ AND NOT }y_2)\textbf{ OR}$ \newline $(x_2\text{ AND NOT }y_1\text{ AND NOT }y_2)\textbf{ OR}$ \newline $(\text{NOT }x_1\text{ AND NOT }x_2\text{ AND }y_1)\textbf{ OR}$ \newline $(\text{NOT }x_1\text{ AND NOT }x_2\text{ AND }y_2)$ \hphantom{mm} &
        Presence of at least one activator and absence of all inhibitors \textbf{or} presence of at least one inhibitor and absence of all activators &
        \cellcolor[HTML]{FF4136} NO & 
        $4$ ($m+k$) \\
    \hline
        $(x_1 \text{ OR }x_2)\text{ AND }(\text{NOT }y_1\text{ OR NOT }y_2)$ \newline \newline “Pairs” (Eq. \ref{eq:pairs}) & 
        $(x_1\text{ AND NOT }y_1)\textbf{ OR }$ \newline $(x_1\text{ AND NOT }y_2)\textbf{ OR }$ \newline $(x_2\text{ AND NOT }y_1)\textbf{ OR }$ \newline $(x_2\text{ AND NOT }y_2)$ \hphantom{mmmmm} &
        Presence of at least one activator \textbf{and} absence of at least one inhibitor &
        \cellcolor[HTML]{2ECC40} YES &
        $4$ ($m \times k$) \\
    \hline
        $(x_1 \text{ OR }x_2)\textbf{ XNOR }(y_1 \text{ OR }y_2)$ &
        $(x_1\text{ AND }y_1)\textbf{ OR }$ \hphantom{mm} \newline $(x_1\text{ AND }y_2)\textbf{ OR }$ \hphantom{mm} \newline $(x_2\text{ AND }y_1)\textbf{ OR }$ \hphantom{mm} \newline $(x_2\text{ AND }y_2)\textbf{ OR }$ \hphantom{mm} \newline \hphantom{m} $(\text{NOT }x_1\text{ AND NOT }x_2\text{ AND}$ \newline $\text{NOT }y_1\text{ AND NOT }y_2)$ \hphantom{mm} &
        Presence of at least one activator and inhibitor pair \textbf{or} absence of all regulators &
        \cellcolor[HTML]{FF4136} NO &
        $5$ ($m \times k + 1$) \\
    \hline
        $True$ when $x_1+x_2 > y_1+y_2$ \newline \newline “Inh-win” (Eq. \ref{eq:thres-inh}) &
        $(x_1\text{ AND }x_2\text{ AND NOT }y_1)\textbf{ OR}$ \hphantom{mm} \newline $(x_1\text{ AND }x_2\text{ AND NOT }y_2)\textbf{ OR}$ \hphantom{mm} \newline $(x_1\text{ AND NOT }y_1\text{ AND NOT }y_2)\textbf{ OR}$ \newline $(x_2\text{ AND NOT }y_1\text{ AND NOT }y_2)$ \hphantom{mm} &
        Number of present activators is larger than the number of present inhibitors &
        \cellcolor[HTML]{2ECC40} YES &
        4 \\
    \hline
        $True$ when $x_1+x_2 \ge y_1+y_2$ \newline \newline “Act-win” (Eq. \ref{eq:thres-act}) & 
        $(x_1\text{ AND }x_2)\textbf{ OR}$ \hphantom{mmm} \newline $(x_1\text{ AND NOT }y_1)\textbf{ OR}$ \newline $(x_1\text{ AND NOT }y_2)\textbf{ OR}$ \newline $(x_2\text{ AND NOT }y_1)\textbf{ OR}$ \newline $(x_2\text{ AND NOT }y_2)$ \hphantom{mmmmm} &
        Number of present activators is larger than or equal to the number of present inhibitors &
        \cellcolor[HTML]{2ECC40} YES &
        5 \\
    \hline
  \end{tabular}}
  \bigskip
  \caption{Several Boolean regulatory functions with four regulators ($m=2$ positive $\{x_1,x_2\}$, $k=2$ negative $\{y_1,y_2\}$) and some metrics are presented. 
  The two first columns provide two different function forms: a standard one, i.e. either the link operator form distinguishing activating and inhibiting regulators or a simple description in the case of the threshold functions, and the CDNF which is a special case of DNF (Section \ref{subsec:consistency}). 
  The “Biological Interpretation” states in words the conditions that make a BRF become $True$, and is explicitly translated from the terms in the corresponding CDNF.
  The “Consistent” column states if the functions satisfy the properties 1-3 from Section \ref{subsec:consistency} (YES, green-colored) or there are inconsistencies with the underlying regulatory structure (NO, red-colored), i.e. if an activator (resp. inhibitor) appears as a negative (resp. positive) literal in the corresponding CDNF.
  The functions are sorted according to an increasing complexity metric (“Complexity” column), which is the number of terms in each respective, minimum-length CDNF expression.
  In parentheses we provide the generalized formula for the number of CDNF terms of the link operator functions with $m$ activators and $k$ inhibitors.}
  \label{tab:table2}
\end{table}

\subsection{DNF links to biological semantics}

We argue here that the DNF is the most adequate function form to help us address the aforementioned issues.
Every Boolean regulatory function expressed in DNF, has a biological characterization that is directly derived from the formula itself: each term in the DNF is an activation condition, i.e. a list of regulators, some present (the positive literals) and some absent (the negative literals), which, when combined, make the target (output of the function) active.
Further merging of all the conditions using OR-semantics into a description of how the regulators influence the target’s output, facilitates the biological interpretation of any Boolean regulatory function.

In Table \ref{tab:table2}, we show a list of BRFs with two positive and two negative regulators. 
Most of the BRFs presented have a different link operator separating the activators from the inhibitors.
Using the functions standard expressions (1st column) makes it very hard to derive a meaningful biological characterization as expressed in the 3rd column of Table \ref{tab:table2}. 
For example, defining a meaningful description of the “NOR” or “NAND-NOT” equations using only their standard expression, is a very difficult task. 
In contrast, by using the equivalent DNFs (2nd column) we can make an explicit, “1-1” correspondence between mathematical formulation and biological interpretation and use it to compare the different functions’ meanings.
Thus, by expressing the “NAND-NOT” equation in DNF, we can precisely identify the conditions that make the outcome of the function $True$ and translate these into a meaningful description such as “Presence of any inhibitor or absence of all activators”. 
Consequently, we are led to a generalized and independent of the number of regulators description of this link operator function. 
Such a description is intuitive to human interpretation and reasoning, in terms of the function’s applicability, e.g. in comparing the “AND-NOT” and “NAND-NOT” biological interpretations, we see that the first is semantically plausible while the second completely contradicts the underlying biology.

\section{Characterizing consistent regulatory functions} \label{sec:consistency}

\subsection{The 3 consistency properties} \label{subsec:consistency}

As observed in Table \ref{tab:table2}, not only can the DNF be used to uncover the biological interpretation of any BRF and subsequently help determine its plausibility, but it also provides a means to compare the different function meanings. 
Still, we need a more refined, technical description that is able to express the implausibility of the “NAND-NOT” or “NOR” cases directly from their mathematical formulas, and which would be applicable to every BRF. 
We define the \textit{consistency} attribute of a BRF to describe its compliance with the underlying regulatory network structure.

The first step in making a Boolean model is to build a graph (PKN), assembling the regulatory entities of interest from various databases or the scientific literature, and use causality information to connect them through their regulatory action on other entities. 
As such, a network structure can be defined, in which entities can regulate (either positively or negatively) some of the other entities. 
Using such a simple network-driven formalization, we define a set of three properties that describe the set of all the \textit{consistent Boolean regulatory functions}, i.e. the functions that comply with the underlying regulatory structure. 
So, for a consistent BRF, the following propositions are satisfied \cite{Cury2019}:

\begin{enumerate}
    \item Its regulators can be partitioned into two disjoint sets: the set of \textit{activators} (positive regulators, enhance target’s activity) and the set of \textit{inhibitors} (negative regulators, suppress target’s activity).
    This stems from the fact that every interaction in the PKN has a fixed sign (either positive or negative).
    As such, there are no dual regulations, i.e. a regulator cannot activate and inhibit a target at the same time.
    This property essentially makes the set of consistent BRFs a subset of the \textit{monotone} Boolean functions \cite{crama2011boolean}.
    \item All regulators are \textit{essential}: for every regulatory input, inverting their values, will also, in at least one configuration of states of other regulators, change the output of the function.
    This means that all regulators are indispensable for deriving the target’s activity output.
    \item A consistent BRF can be represented in a unique \textit{complete} DNF (CDNF) which is also known as Blake’s Canonical Form \cite{Blake1937}.
    This is a consequence of property (1), since monotone Boolean functions expressed in any DNF, can be further simplified by removing redundant literals, resulting in the equivalent unique CDNF expression \cite{crama2011boolean}. 
    This property is really important since it allows us to identify which regulatory entities are activators and which are inhibitors from the corresponding CDNF expression of a consistent BRF: an activator will always appear as a positive literal, whereas an inhibitor will always appear as a negative literal.
\end{enumerate}

We provide an example to delineate the difference between the DNF and CDNF forms and show violations of the consistency properties.
In Table \ref{tab:table3}, we present three Boolean functions, expressing the output of a target regulated by one activator ($x_1$) and one inhibitor ($y_1$).
The functions $f_2$ and $f_3$ are in CDNF whereas $f_1$ is in normal DNF, since the positive regulator $x_1$ appears both as a positive and a negative literal (i.e. it acts as a dual regulator, making $f_1$ inconsistent). 
Notice that $f_1$ reduces to $f_2$ by removing the redundant negative literal ($\lnot x_1$) in the term $(\lnot x_1 \land y_1)$: $y_1$ “absorbs” the larger term and thus a shorter expression manifests, one that covers more $True$ outcomes (i.e. $1$'s) in the truth table. 
In addition, we observe that $f_2$ is inconsistent, since inhibitor $y_1$ appears as a positive literal. 
On the other hand, using a negative literal for inhibitor $y_1$ and a positive one for activator $x_1$, makes $f_3$ consistent.

\begin{table}[h!]
  \centering
  \begin{tabular}{cccccc}
    \toprule
    \multicolumn{2}{c}{Truth Table} & \multicolumn{1}{c}{Term} & \multicolumn{3}{c}{Boolean functions} \\
    \cmidrule(r){1-2}
    \cmidrule(r){3-3}
    \cmidrule(r){4-6}
    $x_1$ & $y_1$ & $(\lnot x_1 \land y_1) $ & $f_1=x_1 \lor (\lnot x_1 \land y_1)$ & $f_2=x_1 \lor y_1$ & $f_3=x_1 \lor \lnot y_1 $ \\
    \midrule
    0&0&0&0&0&1 \\
    0&1&1&1&1&0 \\
    1&0&0&1&1&1 \\
    1&1&0&1&1&1 \\
    \bottomrule
  \end{tabular}
  \bigskip
  \caption{Truth table of three different Boolean regulatory functions with two input regulators, one positive ($x_1$) and one negative ($y_1$). All functions are expressed in DNF. $f_1$ and $f_2$ result in the same target Boolean output, with $f_2$ expressed in CDNF. Activator $x_1$ regulates the target both positively and negatively in $f_1$, making the function non-monotone and thus inconsistent. Inhibitor $y_1$ is a positive literal in $f_2$’s CDNF, making it inconsistent as well. Function $f_3$ is consistent since it’s written in CDNF with the activator $x_1$ and inhibitor $y_1$ appearing as positive and negative literals respectively.}
  \label{tab:table3}
\end{table}

\subsection{Most link operator functions are inconsistent}
 
Examining Table \ref{tab:table2}, we note that the 2nd column presents not just any DNF expression of the studied Boolean regulatory functions, but precisely the CDNF. 
Thus we can immediately identify which BRFs violate at least one of the three properties discussed in Section \ref{subsec:consistency} and are therefore inconsistent with the regulatory topology (this information is presented in the 4th column, labeled “Consistent”). 
Two examples of such inconsistencies include the “NAND-NOT” and “XOR” link operator functions, which have terms in their corresponding CDNF in which an activator $x_i$ appears as a negative literal ($\text{NOT }x_i$) and an inhibitor $y_j$ as a positive literal (as itself). 
In total, from all the BRFs presented in Table \ref{tab:table2}, only the standardized “AND-NOT”, the “OR-NOT”, the “Pairs” and the two threshold functions respect the underlying regulatory topology, as can be verified by examining their respective CDNFs.
The rest of the link operator functions presented are inconsistent and will not be considered for further analysis in this paper.

\section{Truth Density as a measure of expected function output} \label{sec:truth-density}

In this section we present another interesting Boolean function metric, whose properties can be used to add further knowledge about a Boolean function’s behavior. 
This metric, which we call \textit{truth density}, allows us to project what the regulatory target’s output will most likely be when the number of input regulators changes and investigate how the ratio between activators and inhibitors may affect that output.
From a modeler’s perspective, this metric is useful to check if an assigned model parameterization (i.e. use of a specific BRF) can asymptotically predefine the activity state of some targets.
Equipped with this knowledge, a modeler can verify the degree of fitness with the observations that such a parameterization allows, and thus discard a specific function in favor of another, if the latter has a truth density value that better matches the outcome observed in the data.

\subsection{Truth Density}
 
We define the truth density ($TD$) of a Boolean function as the fraction of all input configurations in its corresponding truth table that yield a $True$ ($1$) outcome. 
As such, $TD \in [0,1]$. 
This quantity was first introduced in \cite{kauffman1993} and more recently in \cite{Gherardi2016} under the name of \textit{bias} and was similarly defined as the probability that a Boolean function takes on the value $1$. 
Using the example with the three Boolean functions from Table \ref{tab:table1}, we have $TD_{f_1}=3/8=0.375$, $TD_{f_2}=5/8=0.625$ and $TD_{f_3}=8/8=1$, where the last function is a tautology, with the maximum possible truth density. 
Colloquially, we can say that a Boolean function is \textit{biased}, when it’s truth density is close to $0$ or $1$. 
Since the size of a truth table grows exponentially with the number of inputs of the Boolean function ($n$ inputs correspond to $2^n$ rows), the existence of bias conveys the information that most of the input regulatory configurations result in either an activated or inhibited target (bias towards $1$ or $0$ respectively). 
On the other hand, we shall say that a Boolean function is \textit{balanced}, if it takes on the values $0$ and $1$ equally often, or equivalently, it’s truth density is approximately centered around $1/2$ \cite{Benjamini2005}.

\subsection{Asymptotic truth density results of the consistent regulatory functions}

In Appendix \ref{appx:a}, we present a list of propositions and proofs that provide the exact truth density formulas for the generic forms of the five consistent BRFs we studied in previous sections, namely the “AND-NOT”, “OR-NOT” and “Pairs” link operator functions, and the two threshold functions, “Act-win” and “Inh-win”.
A very important element that enables the straightforward derivation of these formulas, is the use of the equivalent DNF expressions in the proofs, especially for the case of the link operator Boolean functions. 
We also noted that the truth densities of all the aforementioned BRFs depend on two variables: the number of activators and the number of inhibitors (the total number of regulators also appears as a separate variable but it depends on the first two, i.e. it is just their sum). 
Thus, we logically asked if a BRF’s truth density asymptotically tends towards specific values in the $[0,1]$ interval (e.g. the function could be biased or balanced), when the number of its input regulators increases or the ratio between activators and inhibitors changes. 
The results of the asymptotic behavior of the truth density formulas are analytically presented in Appendix \ref{appx:b}.

The asymptotic analysis of the truth density formulas confirmed the intuitive perception that the link operator “AND-NOT” and “OR-NOT” functions show a characteristically opposite behavior with increasing number of regulators: the standardized “AND-NOT” formula depends only on the number of inhibitors and its output tends towards $0$, whereas the “OR-NOT” formula depends only on the number of activators and is biased towards $1$.
On the other hand, the “Pairs” and threshold functions truth densities don't have an asymptotic limit since they depend on both the number of activators and inhibitors.
Therefore, we proceeded in clarifying the role of the \textit{activator-to-inhibitor} ratio by investigating three scenarios which explicitly reveal the functions truth density behavior for a significantly large number of regulators:

\begin{itemize}
    \item A $1:1$ activator-to-inhibitor ratio, where approximately half of the regulators are activators and half are inhibitors.
    \item A high activator-to-inhibitor ratio, where all regulators are activators except one inhibitor.
    \item A low activator-to-inhibitor ratio, where all regulators are inhibitors except one activator.
\end{itemize}

In the $1:1$ ratio scenario, where there is an equal number of activators and inhibitors, the asymptotic behavior of the “AND-NOT” and “OR-NOT” functions corresponds to absolute inhibition ($0$) and activation ($1$) respectively, following the biased behavior shown previously.
The “Pairs” function behaves similarly to the “OR-NOT” function and therefore is also biased towards $1$.
Only the threshold functions show balanced behavior with their truth density value reaching asymptotically $1/2$, since the majority rule does favor neither activators nor inhibitors in this scenario.
On the other hand, in the two extremely unbalanced scenarios, where one set of regulators completely outweighs the other, the asymptotic truth density results of the “AND-NOT” and “OR-NOT” functions depend on each respective scenario.
Specifically, when the inhibitors dominate over the activators, the “OR-NOT” is balanced and the “AND-NOT” is biased, since the former has been shown to depend exclusively on the number of activators (which is just one in this case) for increasingly more regulators, whereas the latter on the number of inhibitors.
Their behavior is reversed when the activators outbalance the inhibitors.
In contrast, the “Pairs” function behaves in a balanced manner, having a truth density asymptotically equal to $1/2$ in both these scenarios, since the single minority regulator is paired with every regulator from the dominant group in the respective DNF expression (Eq. \ref{eq:pairs}) and as a result, it significantly influences the function’s output.
Lastly, the asymptotic results for the threshold functions follow the larger size regulatory group, being biased towards $0$ with significantly more inhibitors and biased towards $1$ with significantly more activators.

\subsection{Validation of asymptotic behavior}

One key issue of immense practical importance for the modeler, which arises when analyzing the asymptotic behavior of the truth density formulas, is the actual number of regulators that effectively make each of the studied functions exhibit the demonstrated behavior.
We noticed that most of the truth density formulas (Eq. \ref{td-asymp:and-not}, \ref{td-asymp:or-not} and \ref{td-asymp:pairs}) are the sum of two to three terms, with only one of them depending exclusively on the number of regulators $n$.
Also, this term is usually $1/2^n$, and can be omitted when considering values larger than $n = 10$ regulators since it’s insignificant ($1/2^{10} \approx 0.001$).
This suggests that the limit value of the truth density formulas may be already derived from a much smaller number of regulators than what is implied by the study of asymptotes.
Therefore, we need to have a more data-centric view of the results from our previous asymptotics analysis of the different BRFs, one that will enable us to verify the mathematically observed behaviors but also identify an approximate range for the number of regulators where the asymptotics decide the outcome of the studied functions.

We generated the complete truth tables for the five consistent BRFs of Table \ref{tab:table2}, from $2$ up to $20$ regulators, accounting for every possible activator-to-inhibitor ratio.
For example, for $n = 10$ regulators, every combination of at least one activator and one inhibitor that adds up to $10$ ($1$ activator + $9$ inhibitors, $2$ activators + $8$ inhibitors, etc.) resulted in a different truth table for each considered Boolean function.
Subsequently, using the generated truth tables, we could easily calculate the exact truth density value for each function at every considered ratio.
The results are shown in Figures \ref{fig:1A} and \ref{fig:1B} for the link operator and threshold functions, respectively.

\clearpage

\begin{figure}[!ht]
    \centering
    
    \begin{subfigure}{0.72\textwidth}
        \centering
        \caption{} \label{fig:1A}
        \includegraphics[width=\textwidth]{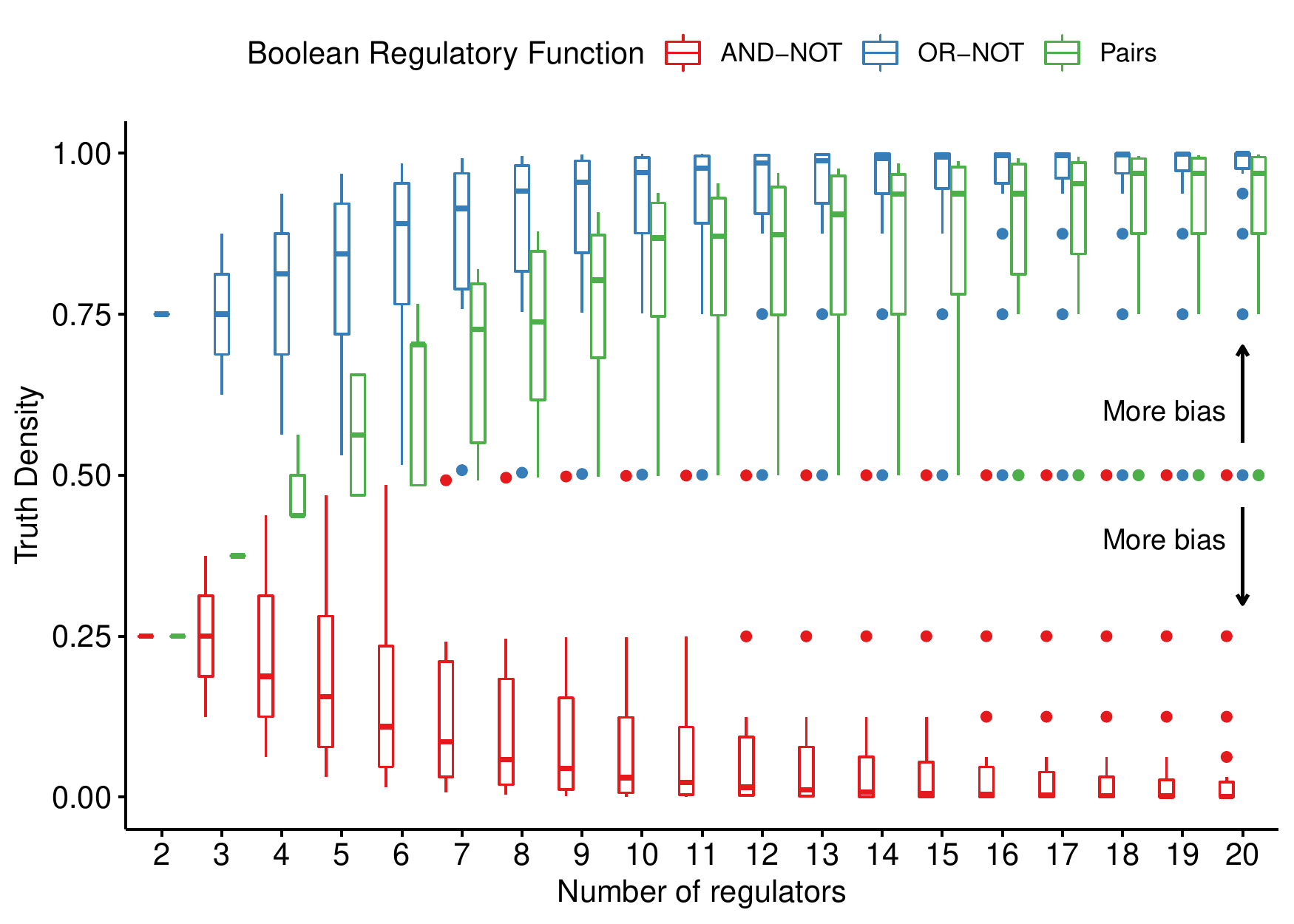}
    \end{subfigure}
    
    \begin{subfigure}{0.72\textwidth}
        \centering
        \caption{} \label{fig:1B}
        \includegraphics[width=\textwidth]{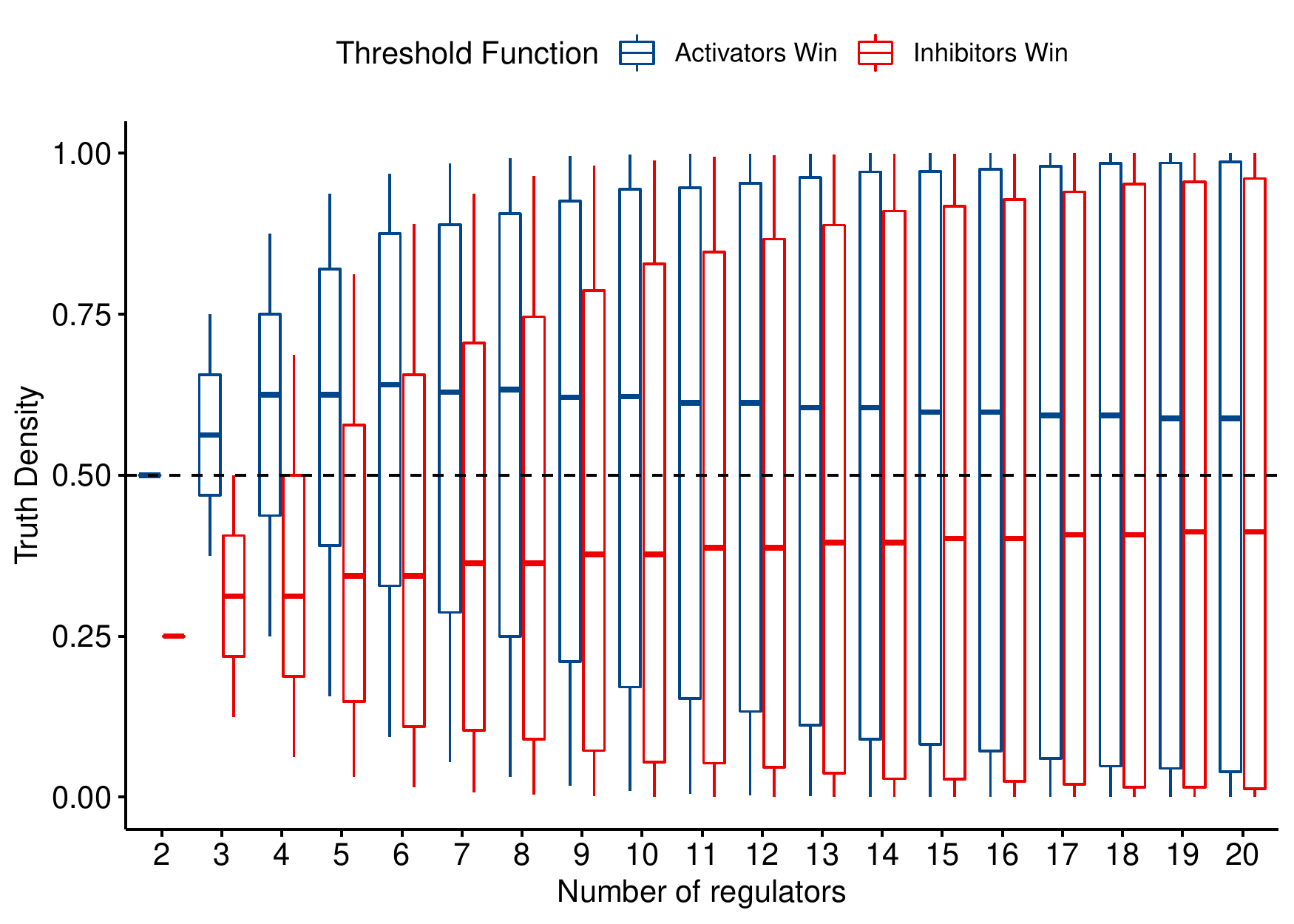}
    \end{subfigure}
    
    \caption{Comparing the truth densities of five different Boolean regulatory functions for different numbers of regulators and activator-to-inhibitor ratios. For each specific number of regulators, every possible combination of at least one activator and one inhibitor that add up to that number, results in a different truth table output with its corresponding truth density value. All such possible configurations up to $20$ regulators are shown. (\textbf{A}) The standardized “AND-NOT” function, along with the “OR-NOT” and “Pairs” functions, show an increasingly biased behavior with more regulators. (\textbf{B}) The two threshold functions “Act-win” and “Inh-win” show a more balanced behavior, since they respect the activator-to-inhibitor ratio and thus demonstrate a larger spectrum of possible truth density values even for higher numbers of regulators.}
    \label{fig:1}
\end{figure}

The data in general shows that the different regulatory functions demonstrate quite dissimilar behaviors with regard to their asymptotic outcome.
In particular, we recapitulate the findings from the asymptotics analysis, namely the bias of the link operator functions, which is evident even from $7$ to $10$ input regulators.
Interestingly, the “Pairs” function follows asymptotically the behavior of the “OR-NOT” function but is in general less biased.
We note that the outliers in Figure \ref{fig:1A} with truth density values closer to $1/2$, represent imbalanced activator-to-inhibitor ratio scenarios, i.e. either considerably more activators than inhibitors for the “AND-NOT” function and the reverse for the “OR-NOT” function, or any imbalanced ratio for the “Pairs” function.
Lastly, Figure \ref{fig:1B} shows that the threshold functions exhibit a more balanced behavior, expressed as a higher spectrum of truth density values for any single number of regulators and with the median truth density asymptotically reaching $1/2$.
This result is due to the fact that threshold functions faithfully follow the activator-to-inhibitor ratio, i.e. with more activators the outcome is biased towards $1$ whereas with more inhibitors the function outcome tends towards $0$.

\section{Link operator parameterization determines activity state in biological networks} \label{sec:network-analysis}

In this section we investigate if a model’s parameterization can effectively decide the activity state of nodes in biological networks.
In more detail, we will use the “AND-NOT” link operator function \cite{Mendoza2006} and its symmetric function “OR-NOT” (Eq. \ref{eq:and-not} and \ref{eq:or-not}), to build Boolean models from prior causal knowledge and check if their activity state profile as determined by dynamic attractor analysis, shows the biased behavior that we observed in Section \ref{sec:truth-density}.

A major motivation for this analysis is the fact that the “AND-NOT” function is extensively used by logical modelers and thus the knowledge of its bias, made possible through the lens of the truth density metric, should be clearly demonstrated in practical use cases, e.g. biological network targets should mostly be in an inhibited state when the “AND-NOT” parameterization is used in their respective Boolean equations and in an active state in the case of the “OR-NOT”.
As such, a modeler could make use of the link operator function bias to select the appropriate model parameterization which statistically guarantees an activity state profile that best matches the one supported by experimental evidence.

\subsection{From topology to link operator Boolean models}

In order to define Boolean models with the “AND-NOT” and “OR-NOT” link operator parameterization forms, we implemented the software \textit{abmlog}, which stands for “\textbf{A}ll possible \textbf{B}oolean \textbf{M}odels \textbf{L}ink \textbf{O}perator \textbf{G}enerator” (\nameref{sec:soft}).
Given a simple interaction (.sif) format file \cite{Toure2020}, representing a PKN with clearly defined, positive and negative causal interactions, the abmlog software outputs all combinatorially possible Boolean models where each link operator equation (deciding the state of a \textit{link operator node}, i.e. one whose Boolean activity state is determined by both positive and negative regulators) will have either the “AND-NOT” or the “OR-NOT” function form.
The models are saved in both the widely-used BoolNet (.bnet) \cite{Mussel2010} format and the gitsbe format \cite{gitsbe-format}, with the latter additionally including the attractors of the Boolean model, calculated via the BioLQM Java library \cite{Naldi2018}.
A simple overview of the software is presented in Figure \ref{fig:2}.

By default, abmlog generates all possible Boolean models with the two link operator parameterizations, the number of which depends on the number of link operator nodes.
For example, if a network has $12$ nodes with both activating and inhibitory regulators, then a total of $2^{12} = 4096$ Boolean models will be generated.
In case the number of all possible Boolean models is very large or space restrictions do not allow the storage of that many models, the software can also be used to generate a random sample of link operator Boolean models from the total parameterization space.
In summary, abmlog is a useful tool that can generate a large pool of Boolean models for subsequent analyses, each with a unique link operator parameterization.

\begin{figure}
    \centering
    \includegraphics{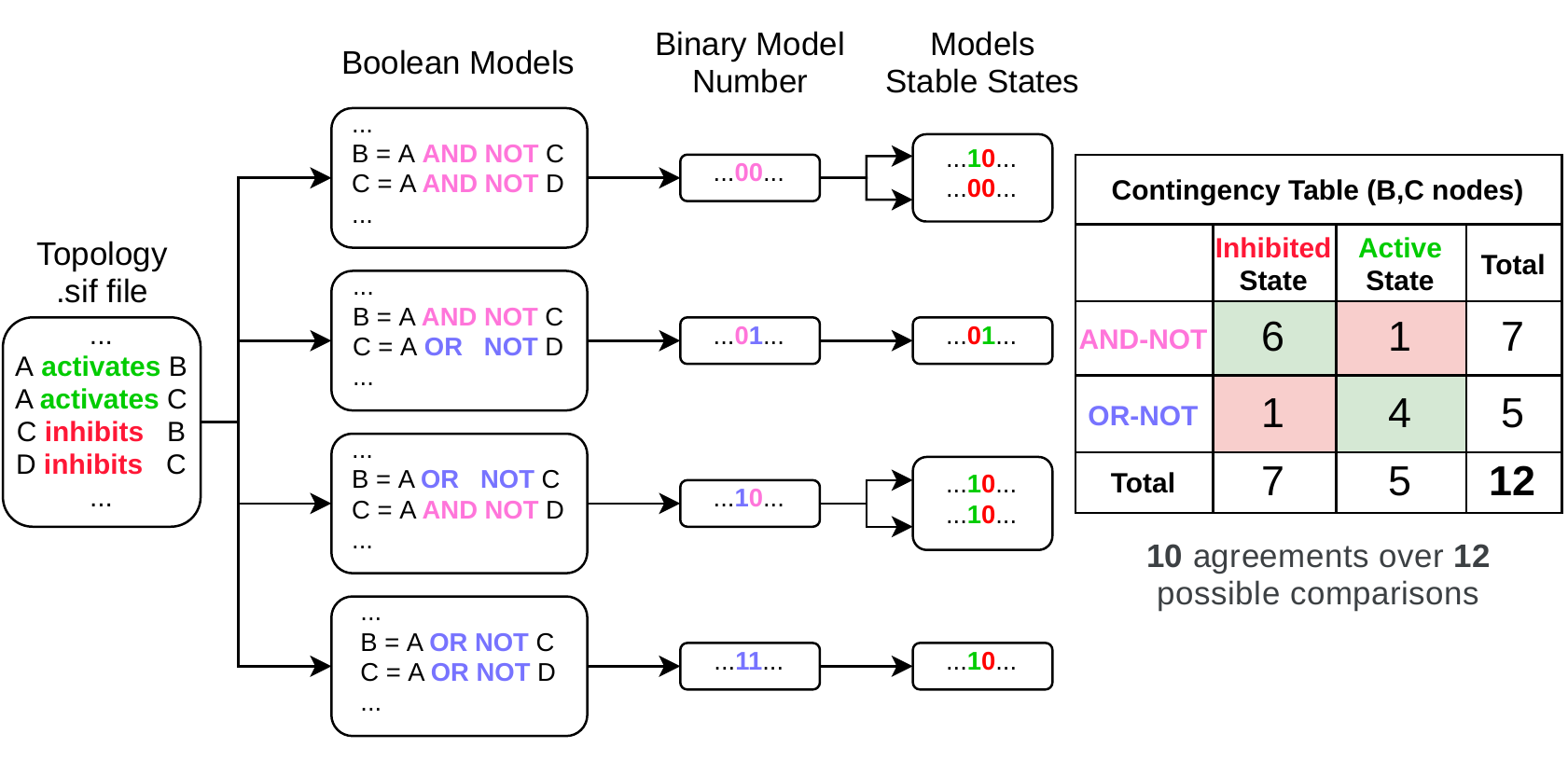}
    \caption{Data-flow overview diagram of the abmlog software and its related contingency table between output model parameterization and stable state activity. A simple interaction file is given as an input to produce a series of Boolean models where equations with both activating and inhibitory regulators have either the “AND-NOT” or the “OR-NOT” formulation. Two link operator equations give rise to a total of $2^2 = 4$ different Boolean models. Each unique parameterization can be represented by a single binary model number, where a “0” corresponds to an equation with the “AND-NOT” link operator and a “1” to an equation with the “OR-NOT”. This representation of parameterization can be directly compared to each of the models' stable states, which enables the creation of a contingency table for the data pertaining to nodes B and C and the derivation of measures of agreement (see Section \ref{subsec:agree}).}
    \label{fig:2}
\end{figure}

\subsection{Measuring agreement between parameterization and stable state} \label{subsec:agree} 

In order to quantify the link operator function bias, we use measures of agreement between parameterization and stable state.
The idea is that the more biased the link operator parameterization is, the higher the expected agreement will be between a target node's link operator assignment and its corresponding stable state.
For the rest of this work, we shall use two measures of agreement, namely the \textit{percent agreement} and Cohen’s \textit{kappa statistic} \cite{Cohen1960}.

In more detail, using the Boolean model data generated by abmlog, we focus in two categorical variables related to a particular node of interest: its link operator parameterization (“AND-NOT”/“0” or “OR-NOT”/“1”) and its corresponding stable state activity (“inhibition” or “activation”), obtained via attractor analysis.
We shall say that these two variables “agree” when a node whose target Boolean equation has the “AND-NOT” link operator (resp. “OR-NOT”) ends up with an inhibited (resp. active) state in the corresponding attractor.
In the case of a Boolean model with multiple attractors, each of the stable states is used separately to measure the agreement between the two aforementioned variables, since the activity of a node might change between the different attractors, but its parameterization stays the same.

To define measures of agreement between the two proposed categorical variables, we visualize their interrelation using a contingency table. 
A total of four data comparison counts can be used to fill in the table's cells: two where the parameterization and stable state match (i.e. node had the “AND-NOT” link operator form and an inhibited stable state or the “OR-NOT” form and an active state) and two where they differ (i.e. node had the “OR-NOT” form and its state was inhibited, or the “AND-NOT” form and an active state).
The percent agreement is then simply defined as the total number of matches divided by the total number of comparisons and is directly interpreted as the percentage of data that the two variables agree upon.
In the example of Figure \ref{fig:2}, the corresponding contingency table counts all the matches and mismatches between the link operator assignments for nodes B and C and their corresponding activity state ($12$ comparisons in total).
Since there are only two mismatches, the percent agreement is equal to $10/12 = 0.83$, meaning that in $83\%$ of the presented data, the link operator parameterization dictated function outcome.
Naturally, a value of $0$ is the absolute minimum score and indicates complete disagreement between the two variables while a perfect agreement score is equal to $1$ or $100\%$.

A more robust statistic that we also apply in the Boolean model data is Cohen’s kappa $(\kappa)$ coefficient \cite{Cohen1960}.
This statistic is used to measure the extent to which data collectors (raters) assign the same score to the same variable (inter-rater reliability) and takes into account the possibility of agreement occurring by chance.
In our case, this can be conceived as one rater that assigns link operator parameterization (“AND-NOT” or “OR-NOT”) and another that assigns stable state activity (“inhibition” or “activation”).
Both variables are converted to a binary outcome ($0$ or $1$), allowing the creation of a contingency table and subsequently the calculation of Cohen’s formula for $\kappa$.
The kappa statistic ranges from $-1$ to $+1$, where a value of $0$ represents the amount of agreement that can be expected from random chance, and a value of $1$ (resp. $-1$) indicates perfect agreement (resp. disagreement) between the raters.
In the example contingency table of Figure \ref{fig:2}, $\kappa = 0.657$, which is a considerable reduction in the level of congruence compared to the $0.83$ percent agreement.

\subsection{Truth Density bias in biological networks}

\subsubsection{Bias guides model parameterization in a cancer signaling network}

We used abmlog on a cancer signaling network, consisting of $77$ nodes and a total of $149$ curated causal interactions that cover a variety of pathways linked to prosurvival and antisurvival cell signaling (e.g. cyclin expression and caspase activation).
This PKN, named CASCADE (\textbf{CA}ncer \textbf{S}ignaling \textbf{CA}usality \textbf{D}atabas\textbf{E}), was successfully used to build a Boolean model able to predict anti-cancer drug combination effects in gastric cell lines \cite{Flobak2015}.
We used the CASCADE version from the Flobak paper (version CASCADE 1.0), with some node naming changes for compatibility with the newest versions \cite{cascade2020}.
The number of nodes with both activating and inhibiting regulators in the CASCADE 1.0 topology is $23$, while the rest of the nodes have regulators that belong to only one of the two regulatory categories.
Thus, using abmlog, we generated all $2^{23}$ possible Boolean models with the “AND-NOT” and “OR-NOT” link operator parameterizations. 
The resulting stable state distribution across all produced models is presented in Figure \ref{fig:3A}.
For our subsequent analysis we will use only the \num{2802224} Boolean models that had exactly one stable state, as it makes the calculation of agreement between a node’s assigned link operator and its corresponding activity state across all the selected models more straightforward.

The agreement results between link operator parameterization and stable state activity across all the selected CASCADE models are presented in Figures \ref{fig:3B} (percent agreement, per node) and \ref{fig:3C} (Cohen’s $\kappa$, nodes with the same number of regulators are grouped together).
The percent agreement results show a high variability across the link operator nodes and range from a minimum of $53\%$ to a perfect agreement ($100\%$).
This suggests that for all nodes, across any selected CASCADE 1.0 Boolean model, there is a higher than random probability that the assignment of the “AND-NOT” (resp. “OR-NOT”) link operator formula in the associated Boolean equations will result in the inhibition (resp. activation) of the target nodes.
So, even though none of the nodes have more than $5$ regulators, we already start seeing signs of the truth density bias in the link operator regulatory functions across a wide range of Boolean models.

When applying Cohen’s $\kappa$ to evaluate level of agreement, we chose a conservative threshold equal to $0.6$, corresponding empirically to a substantial level of agreement \cite{Landis1977, McHugh2012}.
We found that $60\%$ ($14$ out of $23$) of the nodes have a $\kappa$ value below the specified threshold.
Our conclusion is that biological networks with higher in-degree nodes (i.e. more than $7-10$ regulators) are needed to properly assess if there is a truly high level of agreement between Boolean parameterization and function state outcome in the case of the link operator regulatory functions, providing thus conclusive proof of their bias (Section \ref{subsec:scale-free}).

\begin{figure}[!ht]
    \centering
    
    \makebox[\linewidth][c]{ 
        \begin{subfigure}{0.55\textwidth}
            \caption{} \label{fig:3A}
            \includegraphics[width=\textwidth]{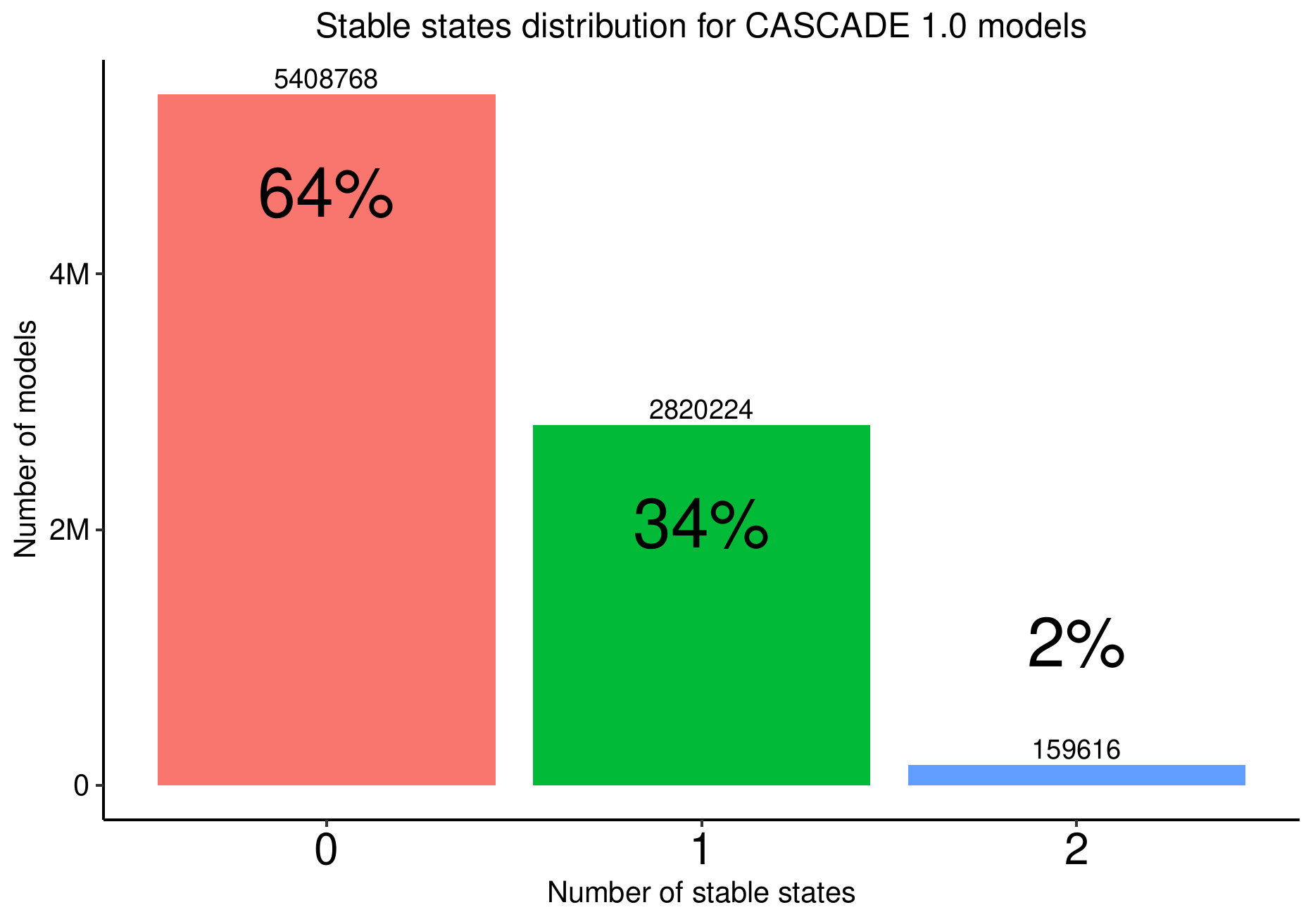}
        \end{subfigure}
        \hfill
        \begin{subfigure}{0.55\textwidth}
            \caption{} \label{fig:3B}
            \includegraphics[width=\textwidth]{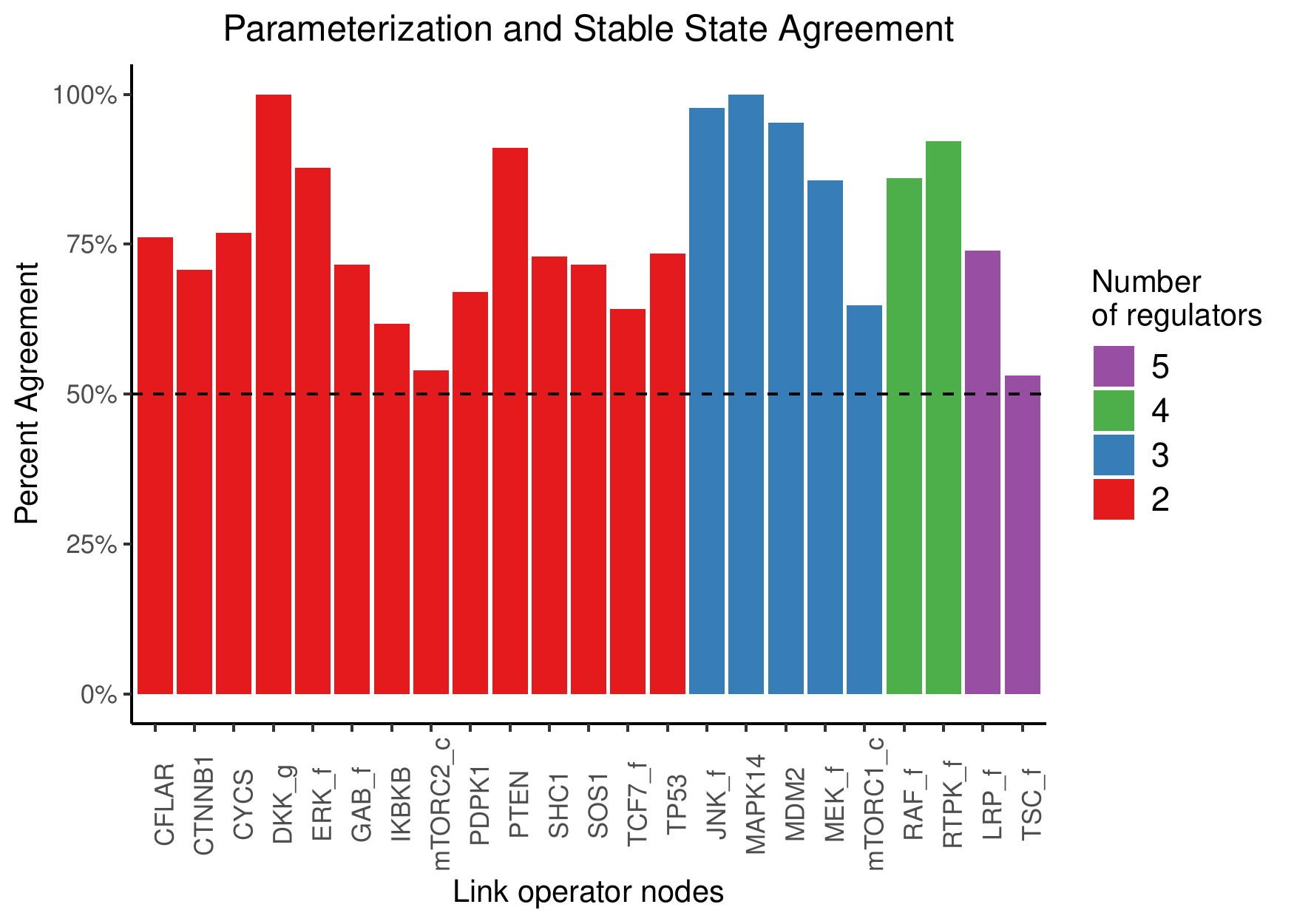}
        \end{subfigure}
    }
    
    \vskip\baselineskip
    
    \makebox[\linewidth][c]{
        \begin{subfigure}{0.55\textwidth}
            \caption{} \label{fig:3C}
            \includegraphics[width=\textwidth]{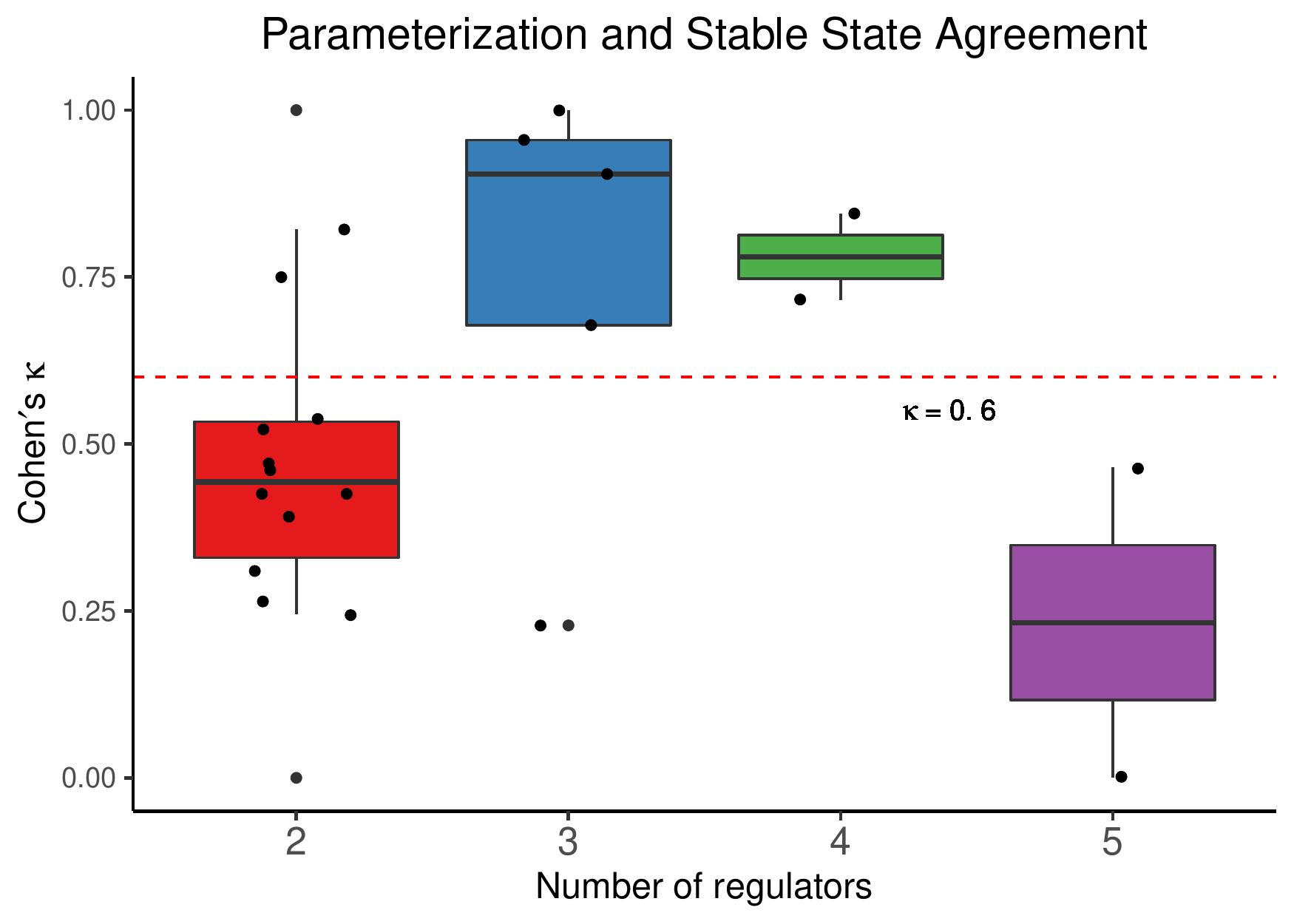}
        \end{subfigure}
        \hfill
        \begin{subfigure}{0.55\textwidth}
            \caption{} \label{fig:3D}
            \includegraphics[width=\textwidth]{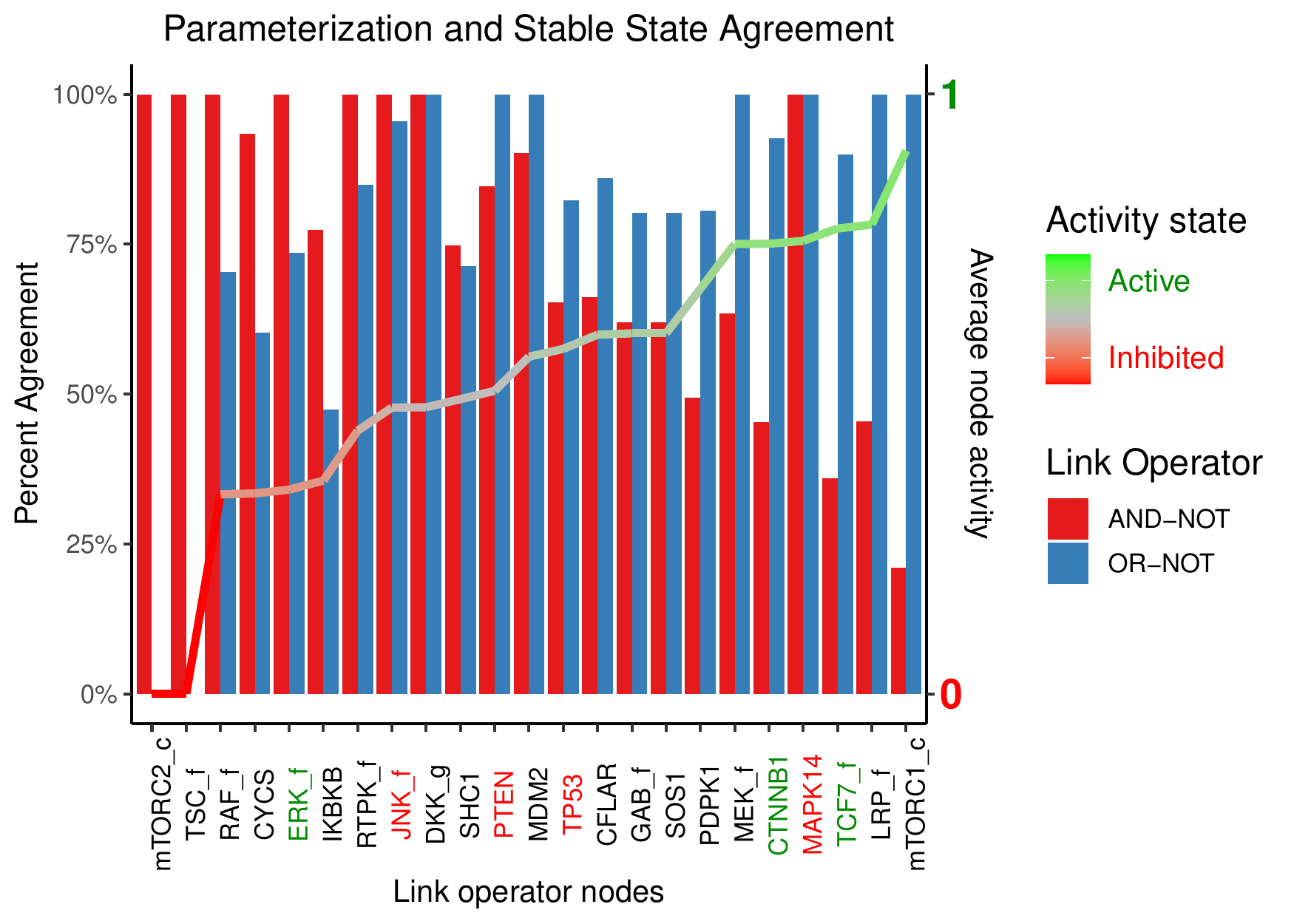}
        \end{subfigure}
    }
    
    \caption{(\textbf{A}) Stable states distribution across all link operator parameterized Boolean models generated by the abmlog software using the CASCADE 1.0 signaling topology. (\textbf{B}) Percent agreement scores between parameterization and activity state across all single stable state CASCADE 1.0 models, for $23$ nodes with both inhibiting and activating regulators. Nodes are sorted according to the total number of input regulators. (\textbf{C}) Same as (\textbf{B}), with the difference that the link operator nodes are now grouped into categories based on the total number of input regulators and Cohen’s $\kappa$ is used as an agreement statistic. (\textbf{D}) Same as (\textbf{B}), with the agreement now calculated as the proportion of matches between a node's link operator and its activity state, in the models that had the specific parameterization. The link operator nodes are sorted according to the average activity state across the considered CASCADE models and the colored node labels indicate literature curated activity profiles from Flobak et al. \cite{Flobak2015}} 
    \label{fig:3}
\end{figure}

Regardless of the presence of bias or not, the agreement results can be used to show how experimental data and topological regulatory knowledge (e.g. the activator-to-inhibitor ratio) can be coupled with the truth density metric to guide the choice of regulatory functions.
In one example scenario, a modeler asks what the most probable link operator parameterization is among the “AND-NOT” and “OR-NOT” forms that matches available experimental evidence. We used a literature curated activity profile derived for the AGS cell line from \cite{Flobak2015}, to annotate $7$ of the link operator nodes in Figure \ref{fig:3D} according to their experimentally validated state (activation or inhibition).
To clearly identify which of the two parameterizations best fits the observed data, for each node we split the CASCADE models in two model pools, representing the “AND-NOT” and “OR-NOT” node parameterizations, and calculated the proportion of models within each pool whose link operator matched the expected state outcome.
For example, in the contingency table of Figure \ref{fig:2}, the equivalent calculation would be to divide the number of matches in each row with the corresponding row total sum, resulting in $6/7 = 85.7\%$ of the “AND-NOT” Boolean equations with an inhibited stable state and $4/5 = 80\%$ of the “OR-NOT” equations with an active target node.
Moreover, the link operator nodes of Figure \ref{fig:3D} are sorted in increasing order by their average stable state activity in the considered CASCADE 1.0 Boolean models.
It is evident that nodes with higher average activity in the stable state have a higher agreement with the “OR-NOT” parameterization whereas nodes with lower average activity, a higher agreement with the “AND-NOT” parameterization ($0.85$ and $-0.74$ Pearson correlation coefficients with $p_{corr}^{\text{\tiny{OR-NOT}}}=2.6\times 10^{-7}$ and $p_{corr}^{\text{\tiny{AND-NOT}}}=5\times 10^{-5}$ respectively, see \nameref{sec:soft}).

More specifically, we observe that for all experimentally validated nodes, a modeler could a priori set the link operator to the appropriate form and get a stable state activation profile that matches the observations (“AND-NOT” to match an inhibition node profile or “OR-NOT” for an activation profile) with a higher probability than if he was randomly choosing one of the two.
For example, the data shows that 90\% of the models with an “OR-NOT” Boolean equation for the target family node \texttt{TCF7\_f}, had the node as active in their respective stable state.
The same is observed for the \texttt{CTNNB1} (92\%) and \texttt{ERK\_f} active nodes (74\%), as well as for the \texttt{TP53} (65\%) and \texttt{PTEN} (85\%) inhibited nodes with the choice of the “AND-NOT” parameterization.
Additionally, all the aforementioned nodes have two regulators (one activator and one inhibitor) and using the respective truth density formulas (Eq. \ref{td:and-not} and \ref{td:or-not}) with $n=2$ and $m=k=1$, we have that $TD_{AND-NOT} = 0.25$ (closer to $0$ or inhibition) and $TD_{OR-NOT} = 0.75$ (closer to $1$ or activation), as was also shown in Figure \ref{fig:1A}.
As such, the nodes observed output matches the statistically expected binary outcomes, showing that even with a low number of regulators, the BRF bias can be used to guide function choice.

In another scenario, a modeler knows that a particular node has a skewed activator-to-inhibitor ratio and wants to exploit such knowledge to make the node conform to a particular activity state of his choice.
A nice example from our data is the family node \texttt{LRP\_f}, with four activators and one inhibitor.
Using the truth density formulas for the two link operator parameterizations (Eq. \ref{td:and-not} and \ref{td:or-not}) with $n=5$, $m=4$ and $k=1$, we have that $TD_{AND-NOT} = 0.47$ and $TD_{OR-NOT} = 0.97$.
So, if the modeler wants to have an active \texttt{LRP\_f} in the stable state, the “OR-NOT” parameterization should be preferred since the “AND-NOT” has an approximate $50\%$ probability for this to happen from a statistical point of view.
These truth density values also match the results from Figure \ref{fig:3D}, since only half of the models that use the “AND-NOT” parameterization end up with an inhibited \texttt{LRP\_f} in the stable state while all of them have an active \texttt{LRP\_f} ($100\%$ agreement) in the case where the “OR-NOT” form is used.
Also, the average activity of \texttt{LRP\_f} across all models is one of the highest in the data, suggesting that imbalanced activator-to-inhibitor ratios could be a direct proxy for predicting regulation outcome.
In a similar situation, but at the other range of the activity spectrum, we have the \texttt{TSC\_f} family node with one activator and four inhibitors.
The truth density values (now using $n=5$, $m=1$ and $k=4$) are $TD_{AND-NOT} = 0.03$ and $TD_{OR-NOT} = 0.53$ respectively.
Therefore, the “AND-NOT” parameterization guarantees the inhibition of the \texttt{TSC\_f} node (data shows $100\%$ agreement) and it should be a modeler’s first choice if that is the desired outcome.
On the other hand, if the activation of \texttt{TSC\_f} was a modeler’s preference, then the choice of the “OR-NOT” form would be the most statistically appropriate according to the truth density metric.
We observe though that there was no model having \texttt{TSC\_f} inhibited in the stable stable, indicating that the complex dynamics of the cancer network can also play a significant role in the function outcome.
In general, we note that the particular configuration of activating and inhibiting regulators of a target in a specific model instance, can influence the dynamics attributable to the parameterization, causing several results from our analysis to differ from the expected behavior of the Boolean functions studied.

\subsubsection{Hub node bias in random scale-free networks} \label{subsec:scale-free}

\vspace{5pt}
In the previous section we showed that the truth density bias can be used to predict regulatory function outcome in a specific cancer signaling network, but the question still remains open for general biological networks.
Also, we found evidence suggesting that Boolean dynamics also plays a significant role in deciding each node’s state in the attractors and in some cases activity state results may contradict what is expected from the use and asymptotic interpretation of the truth density formulas.
Therefore, we now proceed to investigate if networks with higher in-degree nodes (i.e. more input regulators) have stable states that can be unquestionably decided a priori by the truth density metric, using the respective $TD$ formulas for the “AND-NOT” and “OR-NOT” link operator parameterizations.

We study the specific class of scale-free networks \cite{Barabasi1999}, based on the hypothesis that most biological networks exhibit that property, i.e. their node degree distribution follows asymptotically a power law $P(k)\sim k^{-\gamma}$, with $k$ the number of regulators and $\gamma$ the scale-free exponent. 
We note that the CASCADE 1.0 model also exhibits the scale-free property (see \nameref{sec:soft}) and there has been evidence in the literature both in favor and against this hypothesis.
In particular, earlier studies showed that many complex networks (including metabolic ones) are approximately scale-free \cite{Jeong2000, Wuchty2001, Albert2005, Aldana2007}, whereas more recent efforts demonstrated that not all cellular biological networks may share that property \cite{Khanin2006}, but those that do, exhibit the strongest level of evidence of scale-free structure \cite{Broido2019}.
Consequently, we shall use scale-free topologies as acceptable substitutes of real biological networks in our analysis.

The methodology is as follows: we start by generating scale-free topology files with a total of $50$ nodes each and a maximum in-degree $k_{max} = 50$ \cite{Mussel2010}.
For each network, the number of input regulators per node is drawn from a Riemann Zeta distribution with parameter $\gamma$ \cite{Aldana2003}. 
The choice of regulators for each network node, as well as the type of regulation (positive or negative), is uniformly random.
The Zeta distribution allows the creation of in-degree values that far exceed the average connectivity in a network, giving rise to the highest-degree nodes (often called “hubs”), which are the most defining characteristic of the scale-free networks.
The value of the scale-free exponent influences the number of hubs and their in-degree distribution.
More specifically, we created scale-free networks with $\gamma = 2$ and $\gamma = 2.5$, since most of the studied networks have an exponent  between $2$ and $3$ \cite{Aldana2003, Albert2002}.
Comparing the networks built with the above methodology, we found that those with $\gamma = 2$ have more nodes with both activating and inhibiting regulators and higher degree hubs than networks with $\gamma = 2.5$ (Figures \ref{fig:4A} and  \ref{fig:4B}).
These two characteristics suggest that the scale-free networks with $\gamma = 2$ are better suited for use with the abmlog software, since the larger the number of link operator nodes, the more Boolean models can be generated and thus more data comparisons can be made between node parameterization and stable state activity.
Additionally, the presence of higher degree hubs is the perfect testbed for the link operator function bias, which manifests especially for nodes with more than $7-10$ regulators, as we found from our earlier truth density asymptotics analysis (Figure \ref{fig:1A}).

\begin{figure}[!ht]
    \centering
    
    \makebox[\linewidth][c]{ 
        \begin{subfigure}{0.55\textwidth}
            \caption{} \label{fig:4A}
            \includegraphics[width=\textwidth]{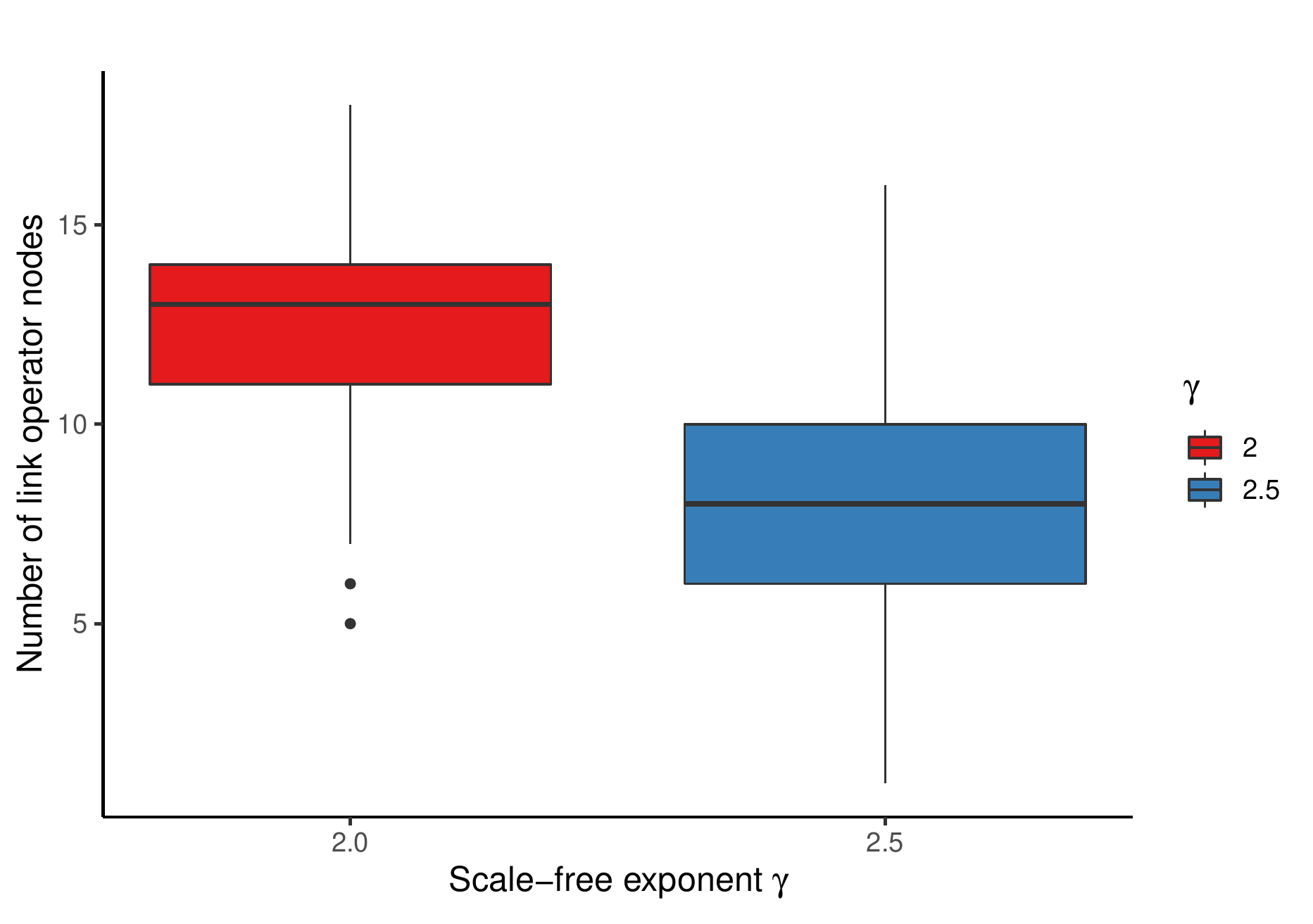}
        \end{subfigure}
        \hfill
        \begin{subfigure}{0.55\textwidth}
            \caption{} \label{fig:4B}
            \includegraphics[width=\textwidth]{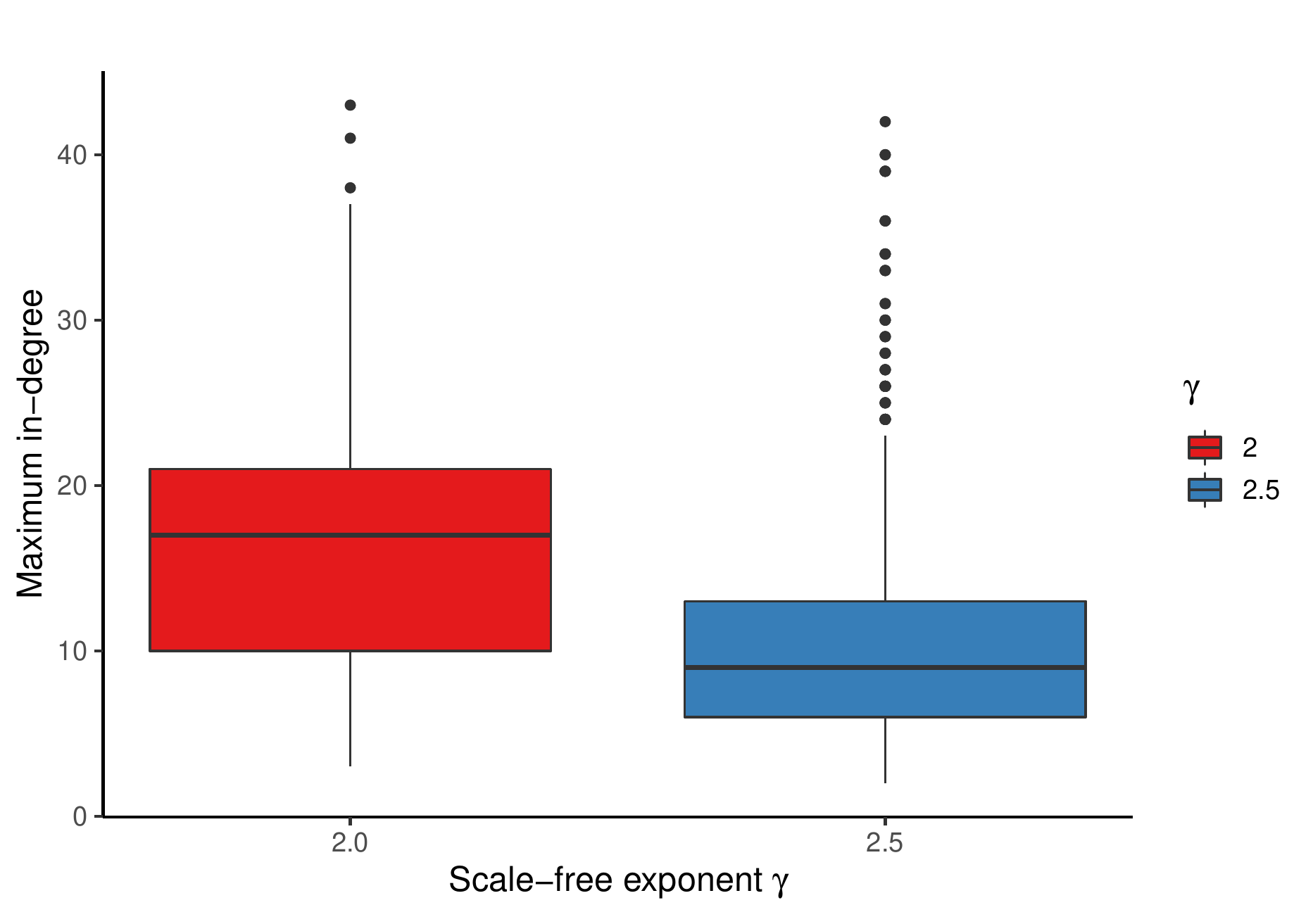}
        \end{subfigure}
    }
    
    \vskip\baselineskip
    
    \makebox[\linewidth][c]{
        \begin{subfigure}{0.55\textwidth}
            \caption{} \label{fig:4C}
            \includegraphics[width=\textwidth]{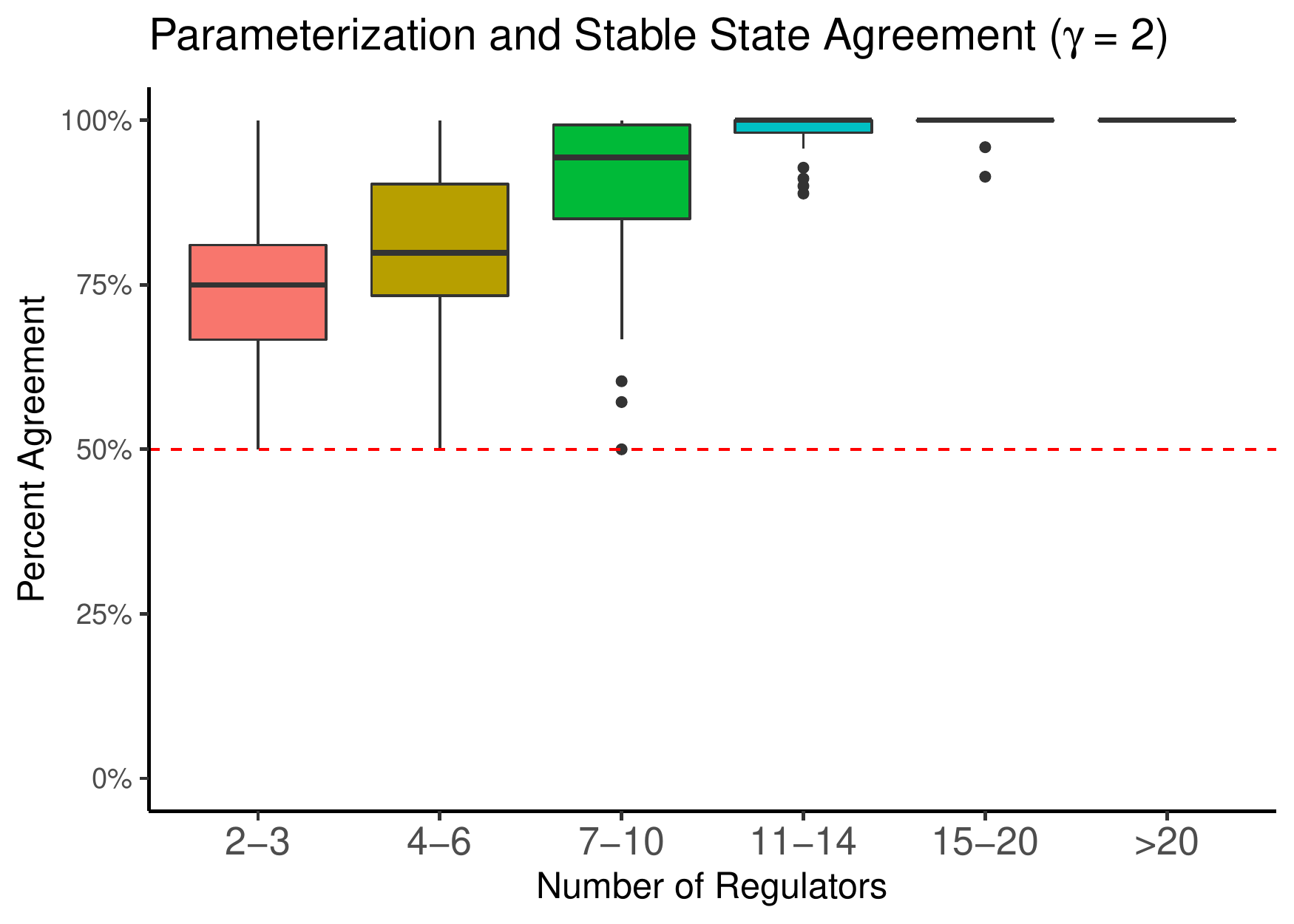}
        \end{subfigure}
        \hfill
        \begin{subfigure}{0.55\textwidth}
            \caption{} \label{fig:4D}
            \includegraphics[width=\textwidth]{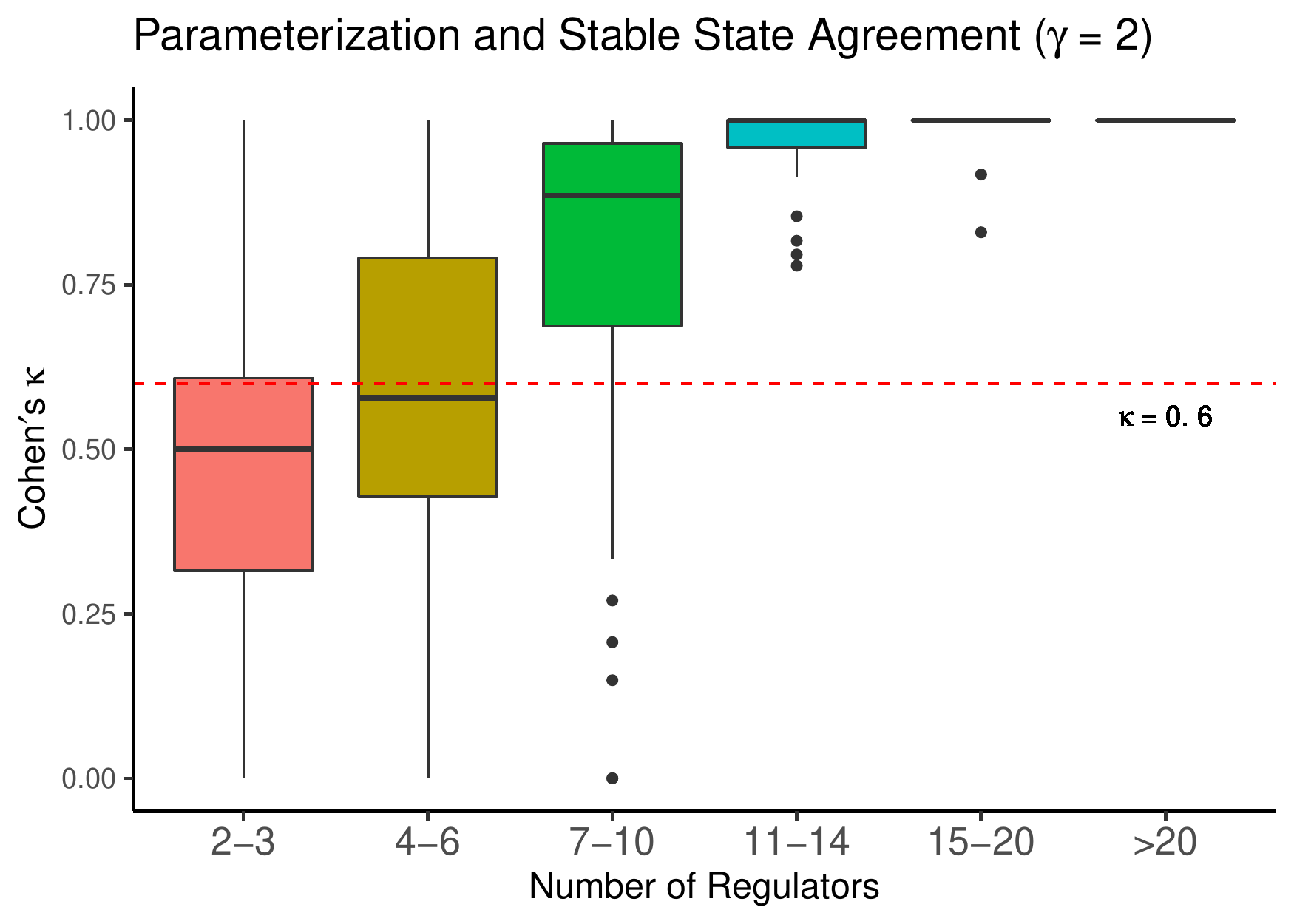}
        \end{subfigure}
    }
    
    \caption{(\textbf{A})-(\textbf{B}) Network statistics for scale-free topologies with different degree exponents. Every network tested has $50$ nodes and a maximum in-degree $k_{max} = 50$. A total of $100$ topologies for $\gamma = 2$ and $1000$ topologies for $\gamma = 2.5$ are compared. Networks with $\gamma = 2$ have a higher median number of nodes with both activating and inhibiting regulators and higher degree hubs. (\textbf{C})-(\textbf{D}) Agreement statistics between link operator parameterization and stable state activity. The data is taken from Boolean models generated with the abmlog software, using scale-free topologies with exponent $\gamma = 2$. A total of $757$ link operator nodes were compared across multiple link operator parameterization configurations with their corresponding stable states. Nodes are grouped in buckets, where each bucket indicates a different range of input regulators. Both the percent agreement and Cohen’s $\kappa$ show considerable congruence between link operator assignment (“AND-NOT” or “OR-NOT”) and resulting stable state (inhibition or activation respectively) for nodes with more than $10$ input regulators.}
    \label{fig:4}
\end{figure}

Our methodology proceeds with using each of the scale-free topologies with $\gamma = 2$ as input to the abmlog software, and generating ensembles of Boolean models parameterized with every possible mix of the “AND-NOT” and “OR-NOT” regulatory functions along with the calculation of their stable states (as demonstrated in Figure \ref{fig:2}).
The produced Boolean models had zero, one, or more stable states.
Interestingly, we observed that around half of the tested scale-free topologies generated Boolean models with no stable states, no matter which combination of link operators was used to define the model parameterization.
Therefore, the randomly assigned regulators, regulatory effects, and Zeta distribution in-degree values, may result in networks which do not have stable phenotypes, suggesting that alternative parameterizations might be more suitable in modeling scenarios which specifically examine stable dynamics.
Nonetheless, we discarded the models with no stable states and used the rest that had single or multiple attractors in our analysis.
Then, for each model node with both activating and inhibiting regulators, we compared its assigned link operator with the activity state value in the corresponding stable state(s), across all the link operator parameterization spectrum that yielded models with stable phenotypes.
The agreement results between parameterization and stable state activity are presented in Figure \ref{fig:4C} for the percent agreement and in Figure \ref{fig:4D} for Cohen’s kappa statistic.

We observe that both presented statistics show a large variation of agreement for nodes with less than $10$ regulators and an increasing agreement with more regulators.
This agreement manifests in link operator nodes parameterized with the “AND-NOT” or “OR-NOT” Boolean functions, while at the same time exhibiting inhibited or active states respectively in the associated model attractors.
Therefore, we conclude that the considered standardized Boolean regulatory functions are biased and their outcomes can be determined a priori from the choice of the corresponding link operator parameterization, especially for nodes with more than $7-10$ regulators.

\section{Discussion} \label{sec:discuss}

The specification of mathematical rules that describe the behavior of biological systems is one of the core aspects of computational modeling.
It is therefore of considerable value to have a list of metrics that can be used to compare different model parameterizations and make an informed decision with regard to the selection of an appropriate regulatory function that better matches the expected behavior in a specific modeling application.

We specifically discussed two characterizations that can assist modelers in comparing various regulatory functions and select the most plausible ones with regard to the causal interaction-based knowledge at hand.
Expressing Boolean functions in DNF makes biological interpretation concrete by explicitly specifying the conditions (presence or absence of the positive and negative regulators, respectively) that make a target active.
Expressing the functions in CDNF allows to easily check for compliance with the underlying regulatory topology and subsequently, the rejection of functions that violate such consistency.
The difference between these two characterizations lies in the fact that the consistency terminology stems from the mathematical world, while biological interpretability is tightly connected to the world of language semantics and thus closer to the modeler's point of view.
Finally, truth density is an informative measure which can be used to verify if the function parameterization dictates biased Boolean outcomes.
It can also be used as a test metric to understand how a function behaves when the number of regulators increases or the balance between the number of activators and inhibitors changes.

Using the truth density metric, we showed the presence of link operator function bias in the hubs of randomly constructed scale-free networks.
A potential application of this finding could be to dramatically decrease the time needed to train Boolean models to fit observations via various optimization methods, by pre-assigning the parameterization of link operator nodes with sufficiently many regulators.
The pruning of the searchable parameterization space, guided by the truth density metric, can result in more efficient automated methods and can enable the training of larger models against data from numerous resources (e.g. large cell line panels).
The hub node bias has also interesting links to the presence of order in biological networks \cite{Kauffman1969}.
The dynamics of a Boolean network can exhibit ordered or chaotic behavior.
Ordered dynamics is characterized by the presence of less stable states and limit cycle attractors with smaller mean length (number of states in a complex attractor) and transition times (number of steps needed to reach an attractor starting out from an arbitrary configuration) \cite{Aldana2003}.
It is also known that the truth density (probability of target expression) as well as the degree exponent $\gamma$ (related to network connectivity) can modulate the dynamic transition between the ordered and chaotic phases.
Moreover, it has been shown that above the critical value of $\gamma_c \sim 2.47$, ordered behavior in the form of stable state dynamics manifests independently of the truth density, whereas for values closer to $\gamma = 2$, order coincides with the presence of high biased nodes (see Fig. 4 in \cite{Aldana2003}).
Our work confirms this phenomenon, since the use of the link operator parameterization guarantees the presence of biased hubs, which enable the scale-free networks to exhibit stability and homogeneity in terms of regulatory output, and thus stay in the ordered dynamic regime.

Searching for other function metrics that are applicable to logical modeling, the \textit{sensitivity} of a Boolean function is one of the most relevant \cite{Shmulevich2004}.
As its name suggests, it measures how sensitive the output of the function is to small changes of its inputs.
Sensitivity is tightly linked to the truth density metric, since a highly homogeneous Boolean function (i.e. a biased one), is unlikely to change its value between similar regulatory input configurations and so, its sensitivity is relatively low.
To compute the average sensitivity value for an arbitrary Boolean function we need to sum over all the \textit{influences} of the input variables, which essentially represent a way to measure individual variable importance.
In the context of regulatory functions, a regulator's influence is defined as the probability that a random toggle on its activity (from active to inactive and vice-versa) will change the value of the Boolean function \cite{Shmulevich2002}.
Therefore, by calculating the influence of every regulator, the modeler can gain knowledge of which ones are more important and control the respective function's outcome.
This transition of perspective from the function level to the regulator level might be advantageous in cases where the modeler's intention is to compare different parameterizations and choose the one for which a particular regulator is labeled as significantly more important than the others, based on the available biological knowledge.

Lastly, an important addition to a universal list of Boolean function metrics for modeling purposes, is the notion of function \textit{complexity}.
A recent definition is given by Gherardi et al. \cite{Gherardi2016}, where the authors defined it as the number of terms in the shortest possible DNF expression of a given Boolean function, divided by the total number of rows in the corresponding truth table.
We presented this information in the last column of Table \ref{tab:table2}, where the BRFs are sorted from lower to higher complexity (note that the CDNF has the minimum number of terms for every BRF included in the table).
One useful observation is that the standardized “AND-NOT” formula \cite{Mendoza2006} is the function with the lowest complexity that is also consistent and thus biologically plausible - all properties that make it a good choice from the modeler's perspective.
Assessing the complexity of the studied regulatory functions using the derived formulas for the minimum number of CDNF terms for any number of activators $m$ and inhibitors $k$ (see last column of Table \ref{tab:table2}), we comment on the fact that all BRFs have very low complexity since $\mathcal{O}(m \times k) \ll 2^{m+k}$, i.e. the number of function terms does not grow as fast as the number of rows in the corresponding truth table.
Same observation has been shown to be true in manually-tuned, experimentally-validated Boolean functions \cite{Gherardi2016}, providing us with another confirmation that the consistent functions from Table \ref{tab:table2} are good candidates for logic-based modeling approaches.

\section{Future work} \label{sec:future}

In this work we make an attempt to address the logical rule specification problem, which can be simply stated as: “Many functions may fit the available observations, which one is the most proper to use?”
Of course what is “proper” can be fairly subjective, but the main point is that a careful consideration of the underlying application context (i.e. what output do I expect in a specific scenario of interest) along with a list of metrics that explicate a Boolean function’s behavior and semantics, provides the user with the appropriate framework to decide on the function parameterization that sets the basis for further model analysis and simulation.
In that regard, interesting directions for further research include the application of the metrics presented in this work in different published biological models, and the subsequent comparison of different regulatory functions within this framework.
Such meta-analyses could potentially indicate regulatory functions that achieve a higher degree of fitness with the observed data or general properties that are common in all Boolean functions used to model biological systems.

An interesting study for example would be to analyze Boolean functions from published biological models that have extreme activator-to-inhibitor ratios.
If such imbalanced ratios also result in proportionally skewed Boolean outcomes (i.e. with more activators, the truth density is closer to $1$ and the reverse with more inhibitors), suggesting that target outcome follows the majority regulatory groups, then the use of threshold functions could be a more proper parameterization alternative, as was shown in Figure \ref{fig:1B}.
Of course, we note that each individual case must be examined with care, since there might be high influence nodes, whose activity defines the target’s output even in the presence of a much larger regulatory group with opposite effects.
For example, \texttt{CASP3} is a biological entity that, when activated, will almost certainly result in the cell’s death even in the presence of a majority of proliferation-positive regulators at any given time.
Subsequently, a more appropriate choice based on the results of this study can be made, either by choosing between the biased functions, which demonstrate a more balanced behavior for such extreme activator-to-inhibitor ratios (e.g. using the “Pairs” or the “AND-NOT” functions which are balanced vs using the “OR-NOT” which would make the target activated most of the time, see \nameref{scenario:2}) or by using refined threshold functions, in which each regulator’s weight will differ in order to match the influence that it has on the target.

There have been only a handful examples of published logical models \cite{Li2004, Davidich2008} and research papers \cite{Chaouiya2013, Wagner1994, Oikonomou2006, Kaderali2009, Greenbury2010, Jack2011, Zanudo2011, Sharan2013} that use the threshold modeling framework in biological systems.
This is partly due to the lack of tools that make threshold functions accessible to the average user, and the availability of such software in open-source environments such as the CoLoMoTo Interactive Notebook \cite{Naldi2018a}.
We believe that the existence of such novel software will enable the construction and configuration of generic Boolean threshold models and provide users of the logical-modeling community and beyond with the necessary toolbox to further study these models.
This will enable applications that depend on the dynamical analysis of Boolean threshold models (identification of attractors, reachability properties, formal verification and control) and the use of optimization methods to calibrate the threshold function parameters to best fit the available experimental data, as is done currently with analytical logic-based functions \cite{Gjerga2020}.

\section*{Software and Data Availability} \label{sec:soft}

The \textit{abmlog} software that was used to generate Boolean models with the “AND-NOT” and “OR-NOT” Boolean regulatory functions is available at \url{https://github.com/druglogics/abmlog} under the MIT License.
We used the version 1.6.0 for this analysis, which is also offered as a standalone package at \url{https://github.com/druglogics/abmlog/packages}.

An extended analysis accompanying the results of this paper is available at \url{https://druglogics.github.io/brf-bias}.
It includes links to the produced model datasets and scripts to reproduce the results and figures of this paper.
In particular, the correlation analysis between average node state in the CASCADE 1.0 models and percent agreement per each link operator is available at \url{https://druglogics.github.io/brf-bias/cascade-1-0-data-analysis.html#node-state-and-percent-agreement-correlation}.
The degree distribution of the CASCADE 1.0 topology and other network statistics are examined in \url{https://druglogics.github.io/brf-bias/cascade-1-0-data-analysis.html#network-properties}.

\section*{Funding}
 
This work was supported by ERACoSysMed Call 1 project COLOSYS (JZ, MK), project UIDB/50021/2020 from Fundação para a Ciência e a Tecnologia - INESC-ID multi-annual funding (PM), the Norwegian University of Science and Technology’s Strategic Research Area ‘NTNU Health’ and The Joint Research Committee between St. Olavs hospital and the Faculty of Medicine and Health Sciences, NTNU - FFU (ÅF).

\textit{Conflict of Interest}: none declared.
 
\section*{Acknowledgments}
 
The authors acknowledge Dr. Vasundra Touré for her contribution in the improvement of Figure \ref{fig:2} and its caption.

\newpage

\appendix

\section{Truth Density formula proofs} \label{appx:a}
For all the following propositions, we consider $f$ to be a Boolean regulatory function $f_{BRF}:\{0,1\}^n \rightarrow \{0,1\}$, with a total of $n$ input regulators separated to two distinct sets, the set of $m \ge 1$ activators $x=\{x_i\}_{i=1}^{m}$ and the set of $k \ge 1$ inhibitors $y=\{y_j\}_{j=1}^{k}$, such that $n = m + k$.

\vspace{5pt}

\begin{proposition}[“AND-NOT” Truth Density] \label{prop:and-not}
The truth density of the “AND-NOT” link operator function $f_{AND-NOT}(x,y) = \left(\bigvee_{i=1}^{m} x_i\right) \land \lnot \left(\bigvee_{j=1}^{k} y_j\right)$, with $m \ge 1$ activators and $k \ge 1$ inhibitors, is given by the formula:
\begin{equation} \label{td:and-not}
    TD_{AND-NOT}=\frac{2^m-1}{2^n} = \frac{1}{2^k}-\frac{1}{2^n}
\end{equation}
\end{proposition}

\begin{proof}
Using the distributive property and De Morgan's law we can express $f_{AND-NOT}$ (Eq. \ref{eq:and-not}) in the equivalent DNF:
\begin{equation*}
\begin{split}
f_{AND-NOT}(x,y) & = \left(\bigvee_{i=1}^{m} x_i\right) \land \lnot \left(\bigvee_{j=1}^{k} y_j\right) \\
                 & = \bigvee_{i=1}^{m} \left( x_i \land \lnot \left( \bigvee_{j=1}^{k} y_j \right) \right) \\ 
                 & = \bigvee_{i=1}^{m} (x_i \land \bigwedge_{j=1}^{k} \lnot y_j) \\ 
                 & = \bigvee_{i=1}^{m} (x_i \land \lnot y_1 \land ... \land \lnot y_k)
\end{split}
\end{equation*}

To calculate $TD_{AND-NOT}$, we need to find the number of rows in $f_{AND-NOT}$'s truth table that result in a $True$ output result and divide that by the total number of rows, which is $2^n$ ($n$ input regulators).

Note that $f_{AND-NOT}$, written in it's equivalent DNF, has exactly $m$ terms.
Each term has a unique $True/False$ assignment of regulators that makes it $True$.
This happens when the activator of the term is $True$ and all of the inhibitors $False$.
Since the condition for the inhibitors is the same regardless of the term we are examining and $f$ is expressed in DNF, the $True$ outcomes of the function $f$ are defined by all logical assignment combinations of the $m$ activators that have at least one of them being $True$ and all inhibitors assigned as $False$.
There are a total of $2^m$ possible $True/False$ logical assignments of the $m$ activators (from all $False$ to all $True$) and $f_{AND-NOT}$ becomes $True$ on all except one of them (i.e. when all activators are $False$), with the corresponding $2^m-1$ truth table rows having all inhibitors assigned as $False$.
Therefore, $TD_{AND-NOT}=(2^m-1)/2^n$.
\end{proof}

\begin{proposition}[“OR-NOT” Truth Density] \label{prop:or-not}
The truth density of the “OR-NOT” link operator function $f_{OR-NOT}(x,y) = \left(\bigvee_{i=1}^{m} x_i\right) \lor \lnot \left(\bigvee_{j=1}^{k} y_j\right)$, with $m \ge 1$ activators and $k \ge 1$ inhibitors, is given by the formula:
\begin{equation} \label{td:or-not}
    TD_{OR-NOT}=\frac{2^n-(2^k-1)}{2^n} = 1 - \frac{1}{2^m} + \frac{1}{2^n}
\end{equation}
\end{proposition}

\begin{proof}
Using De Morgan’s law we can express $f_{OR-NOT}$ (Eq. \ref{eq:or-not}) in the equivalent DNF:
\begin{equation*}
\begin{split}
f_{OR-NOT}(x,y) & = \left(\bigvee_{i=1}^{m} x_i\right) \lor \lnot \left(\bigvee_{j=1}^{k} y_j\right) \\
                & = \left(\bigvee_{i=1}^{m} x_i\right) \lor \left(\bigwedge_{j=1}^{k} \lnot y_j\right) \\
                & = x_1 \lor x_2 \lor ... \lor x_m \lor (\lnot y_1 \land ... \land \lnot y_k)
\end{split}
\end{equation*}

To calculate $TD_{OR-NOT}$, we find the number of rows of $f_{OR-NOT}$'s truth table that result in a $False$ output ($R_{false}$), subtract that number from the total number of rows ($2^n$) to get the rows that result in $f$ being $True$, and then divide by the total number of rows.
As such, $TD_{OR-NOT} = (2^n-R_{false})/2^n$.

Note that $f_{OR-NOT}$, expressed in it's equivalent DNF, has exactly $m+1$ terms.
To make $f_{OR-NOT}$ $False$, we assign the $m$ activators as $False$ and then we investigate which logical assignments of the inhibitors $\{y_j\}_{j=1}^{k}$ make the last DNF term also $False$.
Out of all the possible $2^k$ $True/False$ logical assignments of the $k$ inhibitors (ranging from all $False$ to all $True$) there is only one that does not make the last term of $f_{OR-NOT}$ $False$, which happens specifically when all $k$ inhibitors are $False$.
Therefore, $R_{false}=2^k-1$ and $TD_{OR-NOT}=(2^n-(2^k-1))/2^n$.
\end{proof}

\begin{proposition}[“Pairs” Truth Density] \label{prop:pairs}
The truth density of the “Pairs” link operator function $f_{Pairs}(x,y) = \bigvee_{\forall (i,j)}^{m,k}(x_i\land \lnot y_j)$, with $m \ge 1$ activators and $k \ge 1$ inhibitors, is given by the formula:
\begin{equation} \label{td:pairs}
    TD_{Pairs}=\frac{(2^m-1)(2^k-1)}{2^n}
\end{equation}
\end{proposition}

\begin{proof}
Using the distributive property we can express $f_{Pairs}$ (Eq. \ref{eq:pairs}) in its equivalent conjunction normal form (CNF), where two separate clauses are connected with AND’s ($\land$) and inside the clauses the literals are connected with OR’s ($\lor$): 

\begin{equation} \label{eq:pairs-cnf}
    f_{Pairs}(x,y) = \bigvee_{\forall (i,j)}^{m,k}(x_i\land \lnot y_j) = \left(\bigvee_{i=1}^{m} x_i\right) \land \left(\bigvee_{j=1}^{k} \lnot y_j\right)
\end{equation}

To calculate $TD_{Pairs}$, based on its given CNF, we find the number of rows in its truth table that have at least one $True$ activator ($R_{act}$) and subtract from these the rows in which all inhibitors are $True$ ($R_{inh}$).
Therefore, only the rows that have at least one $True$ activator and at least one $False$ inhibitor will be left, corresponding to the biological interpretation of $f_{Pairs}$. As such, $TD_{Pairs}=(R_{act}-R_{inh})/2^n$.

$R_{act}$ can be found by subtracting from the total number of rows ($2^n$), the rows that have all activators as $False$.
The number of these rows depends on the number of inhibitors, since for each one of the total possible $2^k$ $True/False$ logical assignments of the $k$ inhibitors (ranging from all $False$ to all $True$), there will be a row in the truth table with all activators as $False$.
Therefore, $R_{act} = 2^n - 2^k = 2^{m+k} - 2^k = 2^k(2^m-1)$.

$R_{inh}$ depends on the number of activators, since for each one of the total possible $2^m$ $True/False$ logical assignments of the $m$ activators (ranging from all $False$ to all $True$), there will be a row in the truth table with all inhibitors as $True$.
Note that we have to exclude one row from this result, which is exactly the row that has all activators as $False$ since it's not included in the $R_{act}$ rows.
Therefore, $R_{inh}=2^m-1$ and $TD_{Pairs}=(R_{act}-R_{inh})/2^n=(2^k(2^m-1)-(2^m-1))/2^n$.
\end{proof}

\begin{proposition}[Threshold functions Truth Density] \label{prop:thres}
The truth density of the Boolean threshold functions “Act-win” (Eq. \ref{eq:thres-act}) and “Inh-win” (Eq. \ref{eq:thres-inh}), with $m \ge 1$ activators and $k \ge 1$ inhibitors, is given by the formula:

\begin{equation} \label{td:thres}
TD_{thres} = \frac{\sum_{i=1}^m \left[ \binom{m}{i} \times \sum_{j=0}^{min(u,k)} \binom{k}{j} \right]}{2^n}
\end{equation}

where $u=i$ or $i-1$, depending on the use of the “Act-win” or “Inh-win” function respectively.
\end{proposition}

\begin{proof}
The truth density formula can be easily derived from the observation that we need to count the number of rows in the respective truth table that have more $True$ activators than $True$ inhibitors.
In the case of the “Act-win” function, we also need to add the rows that have an equal number of $True$ regulators in each respective category.

Firstly, we count all the subset input configurations that have up to $m$ activators assigned to $True$.
These include the partial $True/False$ logical assignments that have either a single $True$ activator, a pair of $True$ activators, a triplet, etc.
This is exactly the term $\sum_{i=1}^m \binom{m}{i}$.
Note that each of these activator input configurations is multiplied by a factor of $2^k$ in the truth table to make \textit{complete} rows, i.e. rows where the activators logical assignments stay unchanged and the inhibitor values range from all $False$ to all $True$.
Therefore, we need to specify exactly which inhibitor logical assignments are appropriate for each activator subset input configuration.
To do that, we multiply the size of each activator subset $\binom{m}{i}$ with the number of configurations that have less $True$ inhibitors, i.e. $\sum_{j=0}^{i-1} \binom{k}{j}$.

Let's consider an example with $m,k > 2$ and set $i=2$. 
We find that the number of subsets with $2$ $True$ activators is $\binom{m}{2}$.
Next, we multiply by the number of configurations that have one or no $True$ inhibitors, i.e. $\sum_{j=0}^{1} \binom{k}{j}$. 
This results in the number of rows of interest for the “Inh-win” function, i.e. the rows where there are exactly $2$ activators assigned to $True$ and less than $2$ $True$ inhibitors.
For “Act-win”, we have to multiply up to the $True$ inhibitor pairs, i.e. $\sum_{j=0}^{2} \binom{k}{j}$.
In summation, we count the configurations that have exactly $i$ out of $m$ activators assigned to $True$, and for each one, we multiply by the number of cases that have $0$ up to $i$ inhibitors assigned to $True$ to find the respective rows, i.e. $\binom{m}{i} \times \sum_{j=0}^{i} \binom{k}{j}$.
Repeating this calculation for every possible subset of $i$ activators (from $1$ up to all $m$ of them), and summing the rows up, will result in the numerator of the $TD_{thres}$ formula for the “Act-win” function.

Lastly, note that the \textit{largest} inhibitor configuration subset size that we consider, is the minimum value between the current activator subset size ($u=i$ or $i-1$, depending on which threshold function we use) and the total number of inhibitors $k$.
Therefore, we take into account the case where the number of inhibitors is less than the activator subset size, i.e. $k < u$.
This explains the term $min(u,k)$ in the truth density formula and concludes the proof.
\end{proof}

\section{Truth Density asymptotic behavior} \label{appx:b}
We study the asymptotic behavior of the four truth density formulas (Appendix \ref{appx:a}) for a large number of regulators ($n \rightarrow \infty$).
Note that for the calculations involving the two threshold functions, we will only use the truth density formula corresponding to the “Act-win” function (Eq. \ref{td:thres}, with $u=i$), since both functions have similar formulas and therefore, their limiting behavior is analogous. 
The asymptotics results for each regulatory function are as follows:

\begin{enumerate}
    \item
    The “AND-NOT” function truth density (Eq. \ref{td:and-not}) depends only on the number of inhibitors $k$:
    \begin{equation} \label{td-asymp:and-not}
        TD_{AND-NOT} = \frac{1}{2^k}-\frac{1}{2^n} \sim \frac{1}{2^k}
    \end{equation}
    
    For large $k$, it is biased towards $0$:
    
    \begin{equation*}
        TD_{AND-NOT} = \frac{1}{2^k} \xrightarrow{k \rightarrow \infty}0
    \end{equation*}
    
    \item
    The “OR-NOT” function truth density (Eq. \ref{td:or-not}) depends only on the number of activators $m$:
    \begin{equation} \label{td-asymp:or-not}
        TD_{OR-NOT} = 1 - \frac{1}{2^m} + \frac{1}{2^n} \sim 1-\frac{1}{2^m}
    \end{equation}
    
    For large $m$, it is biased towards $1$:
    
    \begin{equation*}
        TD_{OR-NOT} = 1-\frac{1}{2^m} \xrightarrow{m \rightarrow \infty} 1
    \end{equation*}
    
    \item
    The “Pairs” function truth density (Eq. \ref{td:pairs}) depends on both activators and inhibitors:
    \begin{equation} \label{td-asymp:pairs}
        TD_{Pairs} = \frac{(2^m-1)(2^k-1)}{2^n} = \frac{2^n-2^m-2^k+1}{2^n} = 1 -  \frac{2^m+2^k}{2^n} + \frac{1}{2^n} \sim 1 - \frac{1}{2^k} - \frac{1}{2^m}
    \end{equation}
    
    \item
    The threshold functions truth density (Eq. \ref{td:thres}) depends on both $m$ and $k$ variables and does not have a single fixed limit for $n \rightarrow \infty$.
\end{enumerate}

We now focus on the effect of the ratio ($m:k$) between number of activators and inhibitors on the asymptotic truth density values for $n \rightarrow \infty$.
We consider the following three scenarios for each of the Boolean functions:

\paragraph{Scenario 1} \label{scenario:1} A $1:1$ activator-to-inhibitor ratio, where approximately half of the regulators are activators and half are inhibitors, i.e. $m \approx k \approx n/2$ (consider $n$ is even without loss of generality).

\begin{enumerate}
    \item The “AND-NOT” function truth density is biased towards $0$:
    \begin{equation*}
        \text{(Eq. \ref{td-asymp:and-not})} \Rightarrow TD_{AND-NOT} \sim \frac{1}{2^{n/2}} \xrightarrow{n \rightarrow \infty} 0
    \end{equation*}
    
    \item The “OR-NOT” function truth density is biased towards $1$:
    \begin{equation*}
        \text{(Eq. \ref{td-asymp:or-not})} \Rightarrow TD_{OR-NOT} \sim 1 - \frac{1}{2^{n/2}} \xrightarrow{n \rightarrow \infty} 1
    \end{equation*}
    
    \item The “Pairs” function truth density is biased towards $1$:
    \begin{equation*}
        \text{(Eq. \ref{td-asymp:pairs})} \Rightarrow TD_{Pairs} \sim 1 - \frac{1}{2^{n/2}} - \frac{1}{2^{n/2}} \xrightarrow{n \rightarrow \infty} 1
    \end{equation*}
    
    \item The threshold functions truth density is balanced, meaning its limit asymptotically approaches $1/2$.
    \begin{proof}
    We first rewrite the truth density formula substituting $m = k = n/2$:
    \begin{equation*}
        \text{(Eq. \ref{td:thres})} \Rightarrow TD_{thres} = \frac{\sum_{i=1}^{n/2} \left[ \binom{n/2}{i} \times \sum_{j=0}^{min(i,n/2)} \binom{n/2}{j} \right]}{2^n} = \frac{\sum_{i=1}^{n/2} \left[ \binom{n/2}{i} \times \sum_{j=0}^i \binom{n/2}{j} \right]}{2^n}=\frac{N}{2^n}
    \end{equation*}
    
    Next we simplify $N$, by using the notation $z=n/2$ and $\boldsymbol{x}$ as a meta-symbol for $\binom{z}{x}$.
    For example, $\binom{n/2}{1}=\binom{z}{1}=\boldsymbol{1}$.
    $N$ is therefore expressed as:
  
    \begin{equation*}
        N=\boldsymbol{1}(\boldsymbol{0}+\boldsymbol{1})+\boldsymbol{2}(\boldsymbol{0}+\boldsymbol{1}+\boldsymbol{2})+...+\boldsymbol{z}(\boldsymbol{0}+\boldsymbol{1}...+\boldsymbol{z})
    \end{equation*}
  
    Using the symmetry of binomial coefficients: $\binom{z}{x}=\binom{z}{z-x} \sim \boldsymbol{x} =\boldsymbol{z-x}$, we can re-write $N$ as:
  
    \begin{equation*}
        N=(\boldsymbol{z-1})[\boldsymbol{z}+(\boldsymbol{z-1})]+(\boldsymbol{z-2})[\boldsymbol{z}+(\boldsymbol{z-1})+(\boldsymbol{z-2})]+...+\boldsymbol{0}[\boldsymbol{z}+...+\boldsymbol{0}]
    \end{equation*}
  
    Adding the two expressions for $N$ we have that:
  
    \begin{equation*}
        2N=[\boldsymbol{0}+\boldsymbol{1}...+\boldsymbol{z}]^2+\boldsymbol{1}^2+\boldsymbol{2}^2+...+(\boldsymbol{z-1})^2=2^{2z}+\sum_{\boldsymbol{x}=1}^{z-1} \boldsymbol{x}^2
    \end{equation*}
    
    Substituting back $\binom{z}{x}=\boldsymbol{x}$ and $i=x$ (change of index) in expression $N$, we have that the threshold functions truth density is written as: 
    
    \begin{equation*}
        TD_{thres} = \frac{N}{2^{2z}} = \frac{(1/2) \left[2^{2z}+\sum_{i=1}^{z-1} \binom{z}{i}^2 \right]}{2^{2z}}
    \end{equation*}
  
    As $n \rightarrow \infty$ (and hence $z \rightarrow \infty$), the term $\sum_{i=1}^{z-1} \binom{z}{i}^2$ does not grow as fast as $2^{2z}$ - it is smaller by a factor of $\sqrt{\pi z}$ (see answer to Problem 9.18 in \cite{Graham1994}), and so it becomes negligible:
    
    \begin{equation*}
        \lim_{z\to\infty}TD_{thres}=\lim_{z\to\infty}\frac{(1/2)2^{2z}}{2^{2z}}=\frac12
    \end{equation*}
    \end{proof}
\end{enumerate}

\paragraph{Scenario 2} \label{scenario:2} A high activator-to-inhibitor ratio ($n-1:1$), where all regulators are activators except one inhibitor, i.e. $m = n - 1, k = 1$.

\begin{enumerate}
    \item The “AND-NOT” function truth density is balanced:
    \begin{equation*}
        \text{(Eq. \ref{td-asymp:and-not})} \Rightarrow TD_{AND-NOT} \sim \frac{1}{2^1} = \frac12
    \end{equation*}
    
    \item The “OR-NOT” function truth density is biased towards $1$:
    \begin{equation*}
        \text{(Eq. \ref{td-asymp:or-not})} \Rightarrow TD_{OR-NOT} \sim 1 - \frac{1}{2^{n-1}} \xrightarrow{n \rightarrow \infty} 1
    \end{equation*}
    
    \item The “Pairs” function truth density is balanced:
    \begin{equation*}
        \text{(Eq. \ref{td-asymp:pairs})} \Rightarrow TD_{Pairs} \sim 1 - \frac{1}{2^1} - \frac{1}{2^{n-1}} \xrightarrow{n \rightarrow \infty} \frac12
    \end{equation*}

    \item The threshold functions truth density is biased towards $1$:
    \begin{align*}
        \text{(Eq. \ref{td:thres})} \Rightarrow TD_{thres} &= \frac{\sum_{i=1}^{n-1} \left[ \binom{n-1}{i} \times \sum_{j=0}^{min(i,1)} \binom{1}{j} \right]}{2^n} = \frac{\sum_{i=1}^{n-1} \left[ \binom{n-1}{i} \times \sum_{j=0}^{1} \binom{1}{j} \right]}{2^n}  \\
        &= \frac{\sum_{i=1}^{n-1} \binom{n-1}{i} \times 2}{2^n} = \frac{2^{n-1} - 1}{2^{n-1}} = 1 - \frac{1}{2^{n-1}} \xrightarrow{n \rightarrow \infty} 1
    \end{align*}
\end{enumerate}

\paragraph{Scenario 3} \label{scenario:3} A low activator-to-inhibitor ratio ($1:n-1$), where all regulators are inhibitors except one activator, i.e. $m = 1, k = n - 1$.

\begin{enumerate}
    \item The “AND-NOT” function truth density is biased towards $0$:
    \begin{equation*}
        \text{(Eq. \ref{td-asymp:and-not})} \Rightarrow TD_{AND-NOT} \sim \frac{1}{2^{n-1}} \xrightarrow{n \rightarrow \infty} 0
    \end{equation*}
    
    \item The “OR-NOT” function truth density is balanced:
    \begin{equation*}
        \text{(Eq. \ref{td-asymp:or-not})} \Rightarrow TD_{OR-NOT} \sim 1 - \frac{1}{2^{1}} = \frac12
    \end{equation*}
    
    \item The “Pairs” function truth density is balanced:
    \begin{equation*}
        \text{(Eq. \ref{td-asymp:pairs})} \Rightarrow TD_{Pairs} \sim 1 - \frac{1}{2^{n-1}} - \frac{1}{2^{1}} \xrightarrow{n \rightarrow \infty} \frac12
    \end{equation*}

    \item The threshold functions truth density is biased towards $0$:
    \begin{align*}
        \text{(Eq. \ref{td:thres})} \Rightarrow TD_{thres} &= \frac{\sum_{i=1}^{1} \left[ \binom{1}{i} \times \sum_{j=0}^{min(i,n-1)} \binom{n-1}{j} \right]}{2^n} = \frac{\sum_{j=0}^{min(1,n-1)} \binom{n-1}{j}}{2^n}  \\
        &= \frac{\sum_{j=0}^{1} \binom{n-1}{j}}{2^n} = \frac{1+(n-1)}{2^n} = \frac{n}{2^{n}} \xrightarrow[\text{L'Hôpital Rule}]{n \rightarrow \infty}0
    \end{align*}
\end{enumerate}

\newpage

\bibliography{references}

\end{document}